\documentclass[journal]{IEEEtran}
\IEEEoverridecommandlockouts
\usepackage[T1]{fontenc}% optional T1 font encoding
\usepackage{booktabs}
\usepackage{array}
\usepackage{multirow}
\usepackage{makecell}
\usepackage{tabularx}
\newcolumntype{C}[1]{>{\centering}p{#1}}
\usepackage{cite}
\usepackage{amsmath,amssymb,amsfonts,pifont}
\usepackage{algorithmic}
\usepackage{graphicx}
\usepackage{float}
\usepackage{stfloats}
\usepackage{subfigure}
\usepackage{textcomp}
\usepackage{xcolor}
\usepackage{hyperref}
\usepackage{url}
\usepackage{pifont}
\usepackage{threeparttable}

\allowdisplaybreaks[4]

\usepackage[linesnumbered,ruled,vlined]{algorithm2e}
\SetCommentSty{mycommfont}
\SetKwInput{KwInput}{Input}                % Set the Input
\SetKwInput{KwOutput}{Output}              % set the Output

\usepackage[english]{babel}
\usepackage{amsthm}
\newtheorem{theorem}{Theorem}

\newtheorem{rem}{Remark}

\begin{document}

\title{Enhancing AIGC Service Efficiency with Adaptive Multi-Edge Collaboration in A Distributed System}

\author{Changfu~Xu,~\IEEEmembership{Member,~IEEE,}
Jianxiong~Guo,~\IEEEmembership{Member,~IEEE,}
Jiandian~Zeng,~\IEEEmembership{Member,~IEEE,}
Houming~Qiu,
Tian~Wang,~\IEEEmembership{Senior~Member,~IEEE,}
Xiaowen~Chu,~\IEEEmembership{Fellow,~IEEE,}
and Jiannong~Cao,~\IEEEmembership{Fellow,~IEEE}

\thanks{This work is an extended version of the paper \cite{xu2024enhancing}, which has been published at the 44th IEEE International Conference on Distributed Computing Systems (ICDCS 2024), Jersey City, USA.}
\thanks{This work was supported in part by the Joint Funds of the National Natural Science Foundation of China under Grant U25A20436, the National Natural Science Foundation of China (NSFC) (62372047), the Guangxi Key Research \& Development Program (Gui Ke FN2504240036), the Natural Science Foundation of Guangdong Province (2024A1515011323), the Supplemental Funds for Major Scientific Research Projects of Beijing Normal University, Zhuhai (ZHPT2023002), the Fundamental Research Funds for the Central Universities, Ministry of Education Industry-University Cooperation Collaborative Education Project (240904497110437), the HK RGC Collaborative Research Fund (C5032-23GF), and the NSFC/RGC Collaborative Research Scheme (CRS\_PolyU501/23). (\textit{Corresponding Author: Tian Wang.})}
\thanks{Changfu Xu and Houming Qiu are with the School of Software and IoT Engineering, Jiangxi University of Finance and Economics, Nanchang 330013, China. (E-mail: \{xuchangfu, qiuhouming\}@jxufe.edu.cn)}
\thanks{Jianxiong Guo is with the Institute of Artificial Intelligence and Future Networks, Beijing Normal University, Zhuhai 519087, China, and also with the Guangdong Key Lab of AI and Multi-Modal Data Processing, Beijing Normal-Hong Kong Baptist University, Zhuhai 519087, China. (E-mail: jianxiongguo@bnu.edu.cn)}
\thanks{Jiandian Zeng and Tian Wang are with the Institute of Artificial Intelligence and Future Networks, Beijing Normal University, Zhuhai 519087, China. Tian Wang is also with the College of Computer and Data Science, Fuzhou University, China, and the Computer Science and Mathematics, Fujian University of Technology, 350118, Fuzhou, China. (E-mail: \{jiandian, tianwang\}@bnu.edu.cn)}
\thanks{Xiaowen Chu is with the Hong Kong University of Science and Technology (Guangzhou), Guangzhou, China. (E-mail: xwchu@hkust-gz.edu.cn)}
\thanks{Jiannong Cao is with the Department of Computing, Hong Kong Polytechnic University, Hong Kong. (E-mail: csjcao@comp.polyu.edu.hk)}
}

\maketitle

\begin{abstract}
The Artificial Intelligence Generated Content (AIGC) technique has gained significant traction for producing diverse content. However, existing AIGC services typically operate within a centralized framework, resulting in high response times. To address this issue, we integrate collaborative Mobile Edge Computing (MEC) technology to reduce processing delays for AIGC services. Current collaborative MEC methods primarily support single-server offloading or facilitate interactions among fixed Edge Servers (ESs), limiting flexibility and resource utilization across all ESs to meet the varying computing and networking requirements of AIGC services. We propose AMCoEdge, an adaptive multi-server collaborative MEC approach to enhancing AIGC service efficiency. The AMCoEdge fully utilizes the computing and networking resources across all ESs through adaptive multi-ES selection and dynamic workload allocation, thereby minimizing the offloading make-span of AIGC services. Our design features an online distributed algorithm based on deep reinforcement learning, accompanied by theoretical analyses that confirm an approximate linear time complexity. Simulation results show that our method outperforms state-of-the-art baselines, achieving at least an $11.04\%$ reduction in task offloading make-span and a $44.86\%$ decrease in failure rate. Additionally, we develop a distributed prototype system to implement and evaluate our AMCoEdge method for real AIGC service execution, demonstrating service delays that are $9.23\% - 31.98\%$ lower than the three representative methods.
\end{abstract}

\begin{IEEEkeywords}
Edge computing, Adaptive multi-server collaboration, AIGC, Reinforcement learning, Distributed system
\end{IEEEkeywords}

\section{Introduction}
Due to the super ability of automatic creation (e.g., writing an essay and drawing a picture), the Artificial Intelligence Generated Content (AIGC) technique has gained huge attention in both academia and industry \cite{cao2023comprehensive, lin2024blockchain}. A prime example is ChatGPT \cite{van2023chatgpt}, which has been widely used in many fields. However, current AIGC services have high response times due to their centralized architecture and computation-intensive models \cite{wang2023next,liang2025collaborative}. For instance, users must wait 40-60 seconds to generate an image on the Hugging Face platform \cite{huggingface}.

Recently, collaborative Mobile Edge Computing (MEC) has been proposed to provide low service delays for various computation-intensive Internet of Things (IoT) applications, and substantial approaches have been presented over the past several years \cite{wang2023edge, sun2024edgebrain, ye2024galaxy}. Thus, a natural solution leverages collaborative MEC technology to accelerate AIGC services. Current collaborative MEC methods are typically classified into two types: single-server selection and multi-server selection, each characterized by a different number of servers \cite{fan2024collaborative, xu2023smcoedge, han2021tailored}. In single-server selection methods, when tasks are uploaded to the nearest Edge Server (ES), the scheduler allocates task workloads to the local ES, another ES, or the Cloud Server (CS) for processing. Specifically, if the local ES's computational capacity meets the task's computational demands, the task's workload will be processed locally. Otherwise, two solutions arise: \textit{1)} the whole task workload is further offloaded to another ES \cite{huang2022towards,farhadi2021service} or the CS \cite{zhang2023hybrid, he2023priority} for processing; \textit{2)} a portion of the task workload is further offloaded to another ES \cite{chen2021collaborative} or the CS \cite{zhou2022two} for processing while the remainder is processed at the local ES. 

In contrast, multi-server selection methods exclusively involve collaborative processing among multiple (not limited to two) ESs \cite{xu2024dynamic}. For instance, as depicted in Fig.~\ref{Fig1}, a user submits an image repair task to the ESs running AIGC models for processing. In particular, using collaborative MEC methods (e.g., \cite{xu2024dynamic}), a portion of images is repaired locally at ES A, while the remaining images are allocated to the other two ESs, B and C, for collaborative repair, thereby improving image processing efficiency. However, these multi-server selection methods only offload the task to a fixed number of ESs (e.g., a task is only allocated to predefined $k$ ESs for processing \cite{xu2024dynamic}), which lacks flexibility and limits the exploitation of idle resources on all ESs, consequently diminishing the Quality of Service (QoS) of the system. Therefore, it is desirable to develop a method that utilizes all ESs' idle resources for enhancing AIGC service efficiency.

\begin{figure}[!t]
\centering
\includegraphics[width=\linewidth]{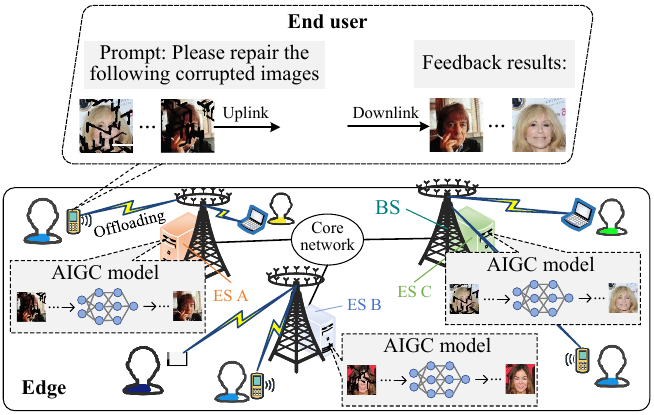}
\caption{An example of AIGC services in collaborative MEC. An end user sends an image-repairing task to its nearest BS. Then, the corrupted images are allocated across three ESs running ACGC models for collaborative repair, reducing processing delays compared to centralized methods.}
\label{Fig1}
\end{figure}

The implementation of the above method, however, poses numerous challenges that warrant further exploration. \textbf{First}, traditional inference services (e.g., object detection) have certain QoS requirements (e.g., accuracy) that are determined by trained network models \cite{zou2025logical}. In contrast, the QoS requirements of AIGC tasks usually vary with user prompt settings. For example, a user can send requests for different quality of generated content. Meanwhile, the AIGC task is also computation-intensive due to its large network model. These characteristics demand varying amounts of computing and networking resources, posing a challenge for allocating them across all ESs while accounting for limited ES resources. \textbf{Second}, MEC systems are typically dynamic, leading to uncertainty about the availability of computing resources at ESs. Thus, in the absence of prior information, determining how to allocate workload is challenging and can delay task offload. \textbf{Third}, the AIGC task typically requires online selection of collaborating ESs and corresponding workload allocation.

Recognizing the above challenges, we first model the offloading problem for Adaptive Multi-server Collaboration in heterogeneous MEC environments, which we call the AM-CMEC problem. Subsequently, we formulate the AM-CMEC problem as an online multivariate Non-Linear Programming (NLP) problem, aiming to minimize the offloading make-span for AIGC tasks. Here, the make-span is determined by the subtasks' longest service delay (including task computing and communication delays) across all ESs. Afterwards, we propose an adaptive multi-edge collaboration method, named \textbf{AMCoEdge}, that solves the AM-CMEC problem in two stages: ES selection and workload allocation. The AMCoEdge fully utilizes the computing and networking resources among all ESs with adaptive multi-ES selection and online workload allocation. Furthermore, we design an Adaptive Deep Q-learning Network (\textbf{ADQN}) model and a Closed-form Workload Allocation (\textbf{CWA}) algorithm to solve the ES selection and workload allocation fraction problems, respectively, achieving near-optimal offloading decisions with approximate linear time complexity. We also develop a distributed edge prototype system to implement and evaluate our AMCoEdge method, validating the AMCoEdge's effectiveness.

This manuscript is a journal extension to our previous conference article \cite{xu2024enhancing}. The primary improvements of this paper are summarized as follows. \textbf{(1)} We propose a Highly Efficient Closed-form Workload Allocation (\textbf{HECWA}) algorithm (seen in Subsection \ref{improved-algorithm}) with corresponding theoretical proof (seen in \textbf{Theorem \ref{effectiveness-Theorem}}) and evaluate the HECWA's cost-effectiveness (Seen in Subsection \ref{cost-effectiveness subsection}). \textbf{(2)} We develop a distributed edge prototype system (seen in Fig. \ref{Fig9}) to implement and evaluate our AMCoEdge method. Also, we include a test-bed experiment (presented in Subsection \ref{test-bed-results}) using real AIGC services to validate AMCoEdge's effectiveness. \textbf{(3)} We provide a detailed discussion (see Section \ref{discussion}) to analyze AMCoEdge's rationale and describe its scalability and limitations, offering insights for future directions in this research field.

The contributions of our work are summarized as follows:
\begin{itemize}
 \item \textbf{Modeling.} We investigate the AM-CMEC problem for providing AIGC services and formulate this problem into an online multivariate NLP problem aimed at minimizing task offloading make-span. We prove the NP-hardness of the offline counterpart of this problem.
 \item \textbf{Method and Algorithm.} We propose an innovative AMCoEdge method that adaptively selects multiple ESs with idle resources to process AIGC tasks, enhancing AIGC service efficiency collaboratively. Meanwhile, we propose an HECWA algorithm to quickly achieve optimal workload allocation decisions, supported by a theoretical proof. We also present theoretical analyses showing that our algorithm achieves near-optimal offloading decisions with approximate linear time complexity.
 \item \textbf{System and Validation.} We develop a prototype system for executing real AIGC services with our AMCoEdge method. Additionally, we evaluate the effectiveness of our AMCoEdge method through comprehensive simulation and test-bed experiments. The simulation results demonstrate that the HECWA algorithm is cost-effective and the AMCoEdge reduces the task offloading make-span and the failure rate by $11.04\%$ and $44.86\%$, respectively, compared to the state-of-the-art baselines. The test-bed results demonstrate $9.23\% - 31.98\%$ fewer service delays than the other three representative methods.
\end{itemize}

The remainder of this paper is organized as follows. Section \ref{sec2} introduces the related work. Section \ref{sec3} gives the system model and problem formulation. Section \ref{sec4} describes the details of the proposed AMCoEdge method and several algorithms with theoretical performance analysis. Sections \ref{evaluation} and \ref{implementation} evaluate the effectiveness of our method. Furthermore, Section \ref{discussion} discusses our method and algorithm. The conclusions are shown in Section \ref{conclusion}.

\section{Related Work}\label{sec2}
Edge computing has significantly amplified the QoS of AIGC services. In this section, we briefly review the literature on cloud-enabled AIGC service, edge-enabled AIGC service, and multi-edge collaboration computing.

\subsubsection{Cloud-enabled AIGC Service}
Currently, AIGC services are enabled with cloud technology. In particular, AIGC services are primarily deployed on the CS, with a centralized framework for client use \cite{midjourney, du2023ai}. For instance, the ChatGPT \cite{van2023chatgpt} application is deployed on the CS by the OpenAI company, where all clients send their requests to the OpenAI CS for centralized processing. Additionally, other AIGC services (e.g., Stable Diffusion (SD) \cite{ho2020denoising} and VideoGPT \cite{maaz2023video} for image generation applications) also use the CS to process clients' requests. However, these AIGC services adopt a centralized framework with limited network bandwidth, requiring substantial computing capacity on the CS. As a result, these services usually have a comparably long response time, especially for many remote clients' requests. 

\subsubsection{Edge-enabled AIGC Service}
Compared with the centralized AIGC framework, edge computing has emerged to provide a lower delay for AIGC services, such as image generation \cite{du2023aigenerated}, next-word prediction \cite{wang2023next}, and text classification \cite{xu2024unleashing}. Most existing work focuses on single-server selection for AIGC. For instance, Du \textit{et al.} \cite{du2024exploring} propose a distributed AIGC method for collaborative MEC, in which AIGC tasks are processed collaboratively across edge devices in wireless networks, thereby improving AIGC QoS. Xu \textit{et al.} \cite{xu2023sparks} present a joint AIGC model caching and inference method in edge networks, reducing offloading latency and utilizing the computing resources of ESs efficiently. Du \textit{et. al.} \cite{du2023aigenerated} propose an AI-generated incentive mechanism and full-duplex semantic communications method that avoids the limited computation power of edge devices for AIGC tasks. However, these methods only allocate task workloads to suitable ESs or CSs for processing, which usually results in low QoS when the task is computation-intensive. 

\subsubsection{Multi-edge Collaboration Computing}
To improve resource utilization in both ES and CS, a partial workload allocation solution for collaborative MEC is further proposed, in which some task workloads are locally processed at the ES and others are further allocated to another ES or the CS for processing \cite{zeng2021coedge}. For example, to reduce the computing delay of task offloading, Li \textit{et al.} \cite{li2020deep} investigate a DRL-based collaborative MEC approach among ESs. However, all these existing methods only support collaborative computing between two ESs or between one ES and the CS to handle task requests. Consequently, the exploitation of idle resources on other ESs will be neglected, thus diminishing the QoS of the entire system. Furthermore, Chu \textit{et al.} \cite{chu2023online} propose a task offloading scheme that uses multiple ESs' computing resources to collaboratively process a task by first partitioning a task into several subtasks offline and then offloading each subtask to a suitable ES for processing. However, this scheme results in disparate processing delays across subtasks across different ESs due to heterogeneous ES resources. As a result, the task processing makespan of this scheme will be extended since the longest processing delay among subtasks determines it. Afterwards, to address the issue of heterogeneous ES resources, Xu \textit{et al.} \cite{xu2024dynamic} design a simultaneous multi-ES offloading method in collaborative MEC. This method simultaneously leverages multiple ESs' idle resources to accelerate task processing by parallel multi-ES selection and dynamic workload partitioning. However, this method only selects a fixed number of ESs, lacking flexibility for heterogeneous tasks. Specifically, when the workload of task requests is low, the task prefers to be processed in the local ES rather than in other ESs, since this reduces transmission delay between the local ES and other ESs. When the workload of task requests is high, the task prefers collaborative processing across all available ESs rather than a fixed number of ESs.

\subsubsection{Comparison Analysis}
 Table I (See Appendix A) summarizes the difference between the current representative methods and our method. Apparently, our method differs from others'. First, current AIGC methods (e.g., \cite{van2023chatgpt}) mainly focus on cloud-enabled AIGC services. In contrast, our method adapts an edge-computing framework to provide AIGC services. Second, although edge-enabled AIGC methods (e.g., \cite{du2024exploring}) have recently been proposed, they allocate a task workload to a single ES for processing, limiting the utilization of other ES resources. However, our method uses multiple ESs to accelerate AIGC task processing. Third, a small number of multi-ES offloading methods (e.g., \cite{xu2024dynamic}) have recently been presented to utilize multiple ESs to collaboratively process tasks simultaneously, but they are only applicable to traditional inference services (e.g., object detection) and only select a fixed number of ESs, limiting the exploitation of idle resources on all ESs. Instead, our method explores adaptive multi-ES offloading to provide AIGC services, selecting idle ESs and optimally allocating task workloads to them for collaborative processing. Also, we propose an HECWA algorithm to quickly achieve optimal workload allocation decisions.

\section{System Model \& Problem Formulation}\label{sec3}
% In this section, we first present the system model of our AM-CMEC problem and then formulate the AM-CMEC problem. Table \ref{table2} summarizes the notations frequently used in this paper.
\subsection{System Model}
Table II (See Appendix A) summarizes the notations frequently used in this paper. Similar to Fig. \ref{Fig1}, we consider AIGC applications in a collaborative MEC system comprising several BSs, each equipped with an ES that deploys the AIGC model. Different ESs have different computing resources. All BSs are connected via a wired core network. BSs provide AIGC services to user devices (UDs) via wireless communication within the BSs' signal range. In a BS $b \in \mathcal{B}$ of our system, as illustrated in Fig. \ref{Fig2}, each ES has a scheduler with a transmission queue and a processing queue. The transmission queue stores task data waiting to be transmitted from the local ES to other ESs.
The processing queue stores tasks waiting to be processed. Let $\mathcal{B} =\{1,\cdots, B\}$ denote the set of ESs or BSs. $b \in \mathcal{B}$ and $B$ denote the $b$-th ES and the total number of ESs, respectively. The system time is split into consecutive time slots and is expressed by $\mathcal{T} =\{1,\cdots, |\mathcal{T}|\}$ \cite{Yu2021when}. $t \in \mathcal{T}$ denotes the $t$-th time slot. The length of the time slot $t$ is denoted as $\Delta$ (in seconds). Next, the task model, decision variable, and make-span model are presented in detail.

\subsubsection{Task Model}
Let $\mathcal{N}_{b,t} = \{1, 2, \cdots, N_{b,t}\}$ denote a set of tasks arrived to BS $b$ at time slot $t$. $n \in \mathcal{N}_{b,t}$ denotes as the $n$-th task. In general, the task $n$ includes the offloading data $d_{n}$ (in bits), the required computation density $\rho_{n}$ (in CPU cycles/bit), and the response deadline $\tau_{n}$ (in seconds). Here, if the task has not been fully processed within $\tau_{n}$ seconds, it will be dropped immediately \cite{tang2022deep}. However, for computationally intensive AIGC tasks executed on GPUs, such as Deep Neural Network (DNN) inference, a more precise computational cost model based on floating-point operations (FLOPs) is necessary. Specifically, the workload of a DNN task $n$ (e.g., ResNet-152 \cite{he2016deep}) is modeled not as $\rho_n \cdot d_n$, but by the total FLOPs across its layers, given by $\sum^{L}_{l=1}2 \cdot J^{\text{height}}_{n,l}J^{\text{width}}_{n,l}C_{n,l}^{\text{in}}C_{n,l}^{\text{out}}M^{\text{height}}M^{\text{width}}$. In this formulation, $J^{\text{height}}_{n,l}$ and $J^{\text{width}}_{n,l}$ denote the height and width of the output feature map at layer $l$, $C_{n,l}^{\text{in}}$ and $C_{n,l}^{\text{out}}$ represent the input and output channels, and $M^{\text{height}}$ and $M^{\text{width}}$ represent the kernel height and width. The subsequent formulations remain consistent for both CPU and GPU workloads.

\subsubsection{Decision Variable}
For each time slot and task, the scheduler decides how many task workload fractions are allocated to all the ESs in the system. Similar to existing works (e.g., \cite{gao2023task, ma2020cooperative, zhou2022two}), let a variable $x_{b,n,t,b'} \in [0, 1]$ denote the workload allocation fraction of task $n$ from the BS $b$ to the ES $b' \in \mathcal{B}$ in time slot $t$. The sum of all the allocation fractions of a task workload should equal 1, i.e.,
\begin{equation}\label{Eq-constrain}
	\sum\nolimits_{b'\in \mathcal{B}}x_{b,n,t,b'}=1,\ \forall b \in \mathcal{B}, n \in \mathcal{N}_{b,t}, t \in \mathcal{T}.
\end{equation}

\begin{figure}[!t]
    \centering
    \includegraphics[width=\linewidth]{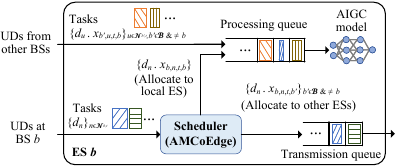}
    \caption{An illustration of the collaborative MEC system model for AIGC services with a BS $b \in \mathcal{B}$. When a task $n$ with data size $d_{n}$ arrives at a BS $b$ at time slot $t$, the scheduler in the BS decides how much task workloads $\{d_{n} \cdot x_{b,n,t,1},\; d_{n} \cdot x_{b,n,t,2},\; \cdots,\; d_{n} \cdot x_{b,n,t,B} \}$ are allocated to which ESs for processing, respectively.}
    \label{Fig2}
    % \vspace{-0.3cm}
\end{figure}

\subsubsection{Make-span Model}
We consider the arrival tasks to be processed on a first-come, first-served basis \cite{tang2022deep}. That is, each arrival task should be processed after its previous arrival task completes in the processing queue. We focus on the processing delay of the task after the BS receives it. Let $T_{b,n,t,b'}^{\text{proc}}$ denote the processing delay of task $n$ that is uploaded to BS $b$ and is offloaded to ES $b'$ in time slot $t$. The $T_{b,n,t,b'}^{\text{proc}}$ includes the task transmission delay over the link from BS $b$ to BS $b'$, the task computing delay on ES $b'$, and the task waiting time in the transmission queue at local BS $b$ and the processing queue at ES $b'$, or equals to the deadline. $T_{b,n,t,b'}^{\text{proc}}$ is formulated as
\begin{align}\label{Eq-process-delay}
	T_{b,n,t,b'}^{\text{proc}} = & \min\{x_{b,n,t,b'} \cdot\left(\frac{d_{n}}{v_{b,b',t}} + \frac{ \rho_{n} \cdot d_{n}}{f_{b'}}\right) \nonumber\\
    & + \mathbb{I}(x_{b,n,t,b'} > 0)\cdot T^{\text{wait}}_{b,n,t,b'},\; \tau_{n}\},
\end{align}
where $\min\{\cdot\}$ denotes the minimum function. $v_{b,b',t}$ (in bits/s) is the transmission rate from BS $b$ to BS $b'$. Here, if the task is offloaded to the local ES, its transmission delay equals 0. $f_{b'}$ (in CPU cycles/s) is the CPU computation capacity of ES $b'$. $\mathbb{I}(\cdot)$ is an indicator function, which equals 1 if $(\cdot)$ is true and 0 otherwise. $ T^{\text{wait}}_{b,n,t,b'}$ is the waiting time of task $n$, which is calculated by 
\begin{equation}\label{waiting-time}
    T^{\text{wait}}_{b,n,t,b'} = \min\left\{\frac{q^{\text{tran,bef}}_{b,n,t}}{v_{b,b',t}} + \frac{q_{t-1,b'}+q^{\text{proc,bef}}_{n,t,b'}}{f_{b'}},\; \tau_{n}\right\},
\end{equation}
where $q^{\text{tran,bef}}_{b,n,t}$ (in bits) is the size of the transmission queue at ES $b$ before sending out task $n$ and $q^{\text{proc,bef}}_{n,t,b'}$ is the length of the processing queue at ES $b'$ before receiving task $n$. They can be achieved by system observation. $q_{t-1,b'}$ (in CPU cycles) is the processing queue workload length of ES $b'$ at the end of time slot $t-1$. The $q_{t,b'} = $
\begin{equation}\label{queue-update}
    \max\{q_{t-1,b'} + \sum_{b \in \mathcal{B}}\sum_{n\in \mathcal{N}_{b,t}} x_{b,n,t,b'} \cdot d_{n} \cdot \rho_{n} - f_{b'} \cdot \Delta,\; 0\},
\end{equation}
where $\max\{\cdot\}$ denotes the maximum function. The task workload of each ES at initial time slot 0 is set to 0, i.e., $q_{0,b'}$ = 0 for $\forall b' \in \mathcal{B}$. The queue data size and workload length of tasks before task $1$ in time slot $0$ are set to 0, i.e., $q^{\text{tran,bef}}_{0,0,b'}$ = 0 and $q^{\text{proc,bef}}_{0,0,b'}$ = 0 for $\forall b' \in \mathcal{B}$.

Furthermore, let $T_{b,n,t}^{\text{make}}$ denote the offloading make-span of task $n$ that arrives at BS $b$ at time slot $t$. Then, the $T_{b,n,t}^{\text{make}}$ is calculated by the longest processing delay of task sub workload at the ESs, i.e.,
\begin{equation}\label{Eq-make-span}
	T_{b,n,t}^{\text{make}} = \max \{T_{b,n,t,b'}^{\text{proc}}\}_{b' \in \mathcal{B}}.
\end{equation}

\subsection{Problem Formulation}
The objective of our AM-CMEC problem is to minimize the sum of all the task offloading make-spans in the whole system, which can be formulated as an online NLP problem:
\begin{align}\label{Eq-AM-CMEC}
	&\min_{\boldsymbol{x}}\lim_{|\mathcal{T}| \rightarrow \infty}\frac{1}{|\mathcal{T}|}\sum\nolimits_{t\in\mathcal{T}} \sum\nolimits_{b\in\mathcal{B}}\sum\nolimits_{n\in\mathcal{N}_{b,t}} T_{b,n,t}^{\text{make}} \\
	&\mbox{s.t.}\quad\text{Eqn.}\; (\ref{Eq-constrain}), \nonumber \\
    &\quad\quad x_{b,n,t,b'} \in [0,1],\; \forall b,b' \in \mathcal{B}, n \in \mathcal{N}_{b,t}, t \in \mathcal{T}.\nonumber
\end{align}
Here, the optimization variable is task workload allocation fraction $\boldsymbol{x} = \{x_{b,n,t,b'}\}_{b,b'\in \mathcal{B}, n\in \mathcal{N}_{b,t},t\in \mathcal{T}}$. This offline counterpart of the AM-CMEC problem is NP-hard as proved by the following \textbf{Theorem \ref{NP-hard-theorem}}.

\begin{theorem}\label{NP-hard-theorem}
The offline counterpart of the AM-CMEC problem is NP-hard.
\end{theorem}
\begin{proof}
    See Appendix C \cite{gu2021layer}.
\end{proof}

\section{Proposed Method and Algorithm}\label{sec4}
% This section describes the details of our proposed method and algorithm. In Subsection \ref{sec4-1}, we present the overall framework of our method and decompose our AM-CMEC problem into two stages: ES selection and workload allocation. Then, we give the specific solving procedures of these two stages in Subsections \ref{sec4-2} and \ref{worload-allocation-section}, respectively. Afterward, we provide the improved algorithm for workload allocation and the implementation of the proposed algorithm with theoretical analyses in Subsections \ref{improved-algorithm} and \ref{sec4-4}, respectively.

\begin{figure*}[!t]
    \centering
    \includegraphics[width=\textwidth]{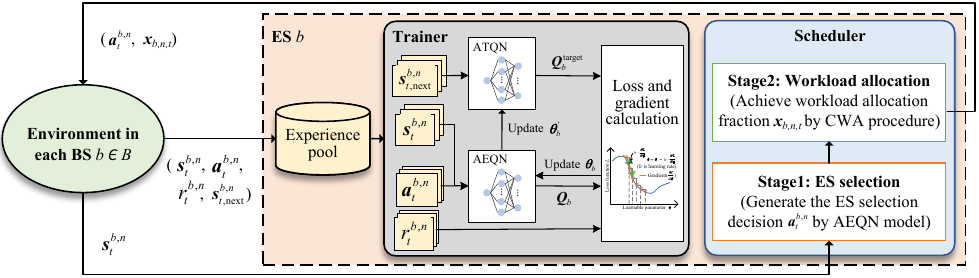}
    \caption{The overall framework of our method in a BS $b \in \mathcal{B}$. When a task $n$ arrives at the BS $b$ at time slot $t$, the scheduler first generates the ES selection decision $\boldsymbol{a}^{b,n}_{t}$ by AEQN model, and then achieves the task workload allocation fraction $\boldsymbol{x}_{b,n,t}$ based on $\boldsymbol{a}^{b,n}_{t}$ by CWA algorithm.}
    \label{Fig3}
\end{figure*}

\subsection{Method Framework}\label{sec4-1}
Driven by the NP-hardness of our problem, we propose a two-stage method to solve the workload allocation decision $\boldsymbol{x}_{b,n,t}$ = $\{x_{b,n,t,b'}\}_{b' \in \mathcal{B}}$. The overall framework is shown in Fig.~\ref{Fig3}. Each BS deploys a trainer and a scheduler in a distributed way. The scheduler includes two stages: ES selection and workload allocation. In the first stage, our method adaptively selects which ESs to collaborate on the AIGC task. Then, in the second stage, our method determines the task-workload fractions allocated to the selected ESs based on the information available for each ES. The details of ES selection and workload allocation are presented as follows.
 
\subsection{Adaptive DQN Design for ES Selection}\label{sec4-2}
Considering heuristic methods (e.g., \cite{lin2024game, sthapit2018computational}) rely on greedy optimization, which easily leads to suboptimal solutions, we design an ADQN model for ES Selection based on the Deep Q-value Network (DQN) model \cite{mnih2015human}. The ADQN model has an Adaptive Evaluation Q-learning Network (AEQN) with parameter $\boldsymbol{\theta}_b$ and an Adaptive Target Q-learning Network (ATQN) with parameter $\boldsymbol{\theta}'_b$. They are trained using historical sample data extracted from the experience pool by the trainer. Notably, unlike the traditional DQN model, which selects only one ES for task offloading, our ADQN model adaptively selects multiple ESs with available resources for each task. Thus, the ADQN model determines the optimal number of multi-ES offloads to ensure that tasks are processed in sequence. The state, action, and reward components of the ADQN model are presented below. 

\textbf{State space.} In our method, at a time slot $t$, each task $n \in \mathcal{N}_{b,t}$ is sequentially processed with one-by-one. Let a vector $\boldsymbol{s}^{b,n}_{t}$ denote the system state when a task $n$ arrives at the BS $b$ at time slot $t$. Intuitively, the size of the arrival task and the length of the task workload at each ES will affect the decision variable and service delay of task offloading. Thus, $\boldsymbol{s}^{b,n}_{t}$ is formulated as
\begin{equation}
	\boldsymbol{s}^{b,n}_{t} = [d_{n}, \boldsymbol{q}_{t-1}],
\end{equation}
where $\boldsymbol{q}_{t-1}=[q_{t-1,1}, q_{t-1,2}, \cdots, q_{t-1,B}]$ are achieved by the Eqn. (\ref{queue-update}). Furthermore, let a variable $\boldsymbol{s}^{b,n}_{t,\text{next}}$ denote the next system state after offloading a task $n$. Since multiple tasks arrive simultaneously at a time slot $t$, the system should allocate the task workload after generating the workload allocation decisions of all tasks arriving in the time slot $t$. Moreover, the queue length $\boldsymbol{q}_t$ will be updated after all tasks arriving in time slot $t$ have been fully assigned for service. Then, the $\boldsymbol{s}^{b,n}_{t, \text{next}}$ is formulated as
\begin{equation}
    \boldsymbol{s}^{b,n}_{t,\text{next}} = \begin{cases}
        \boldsymbol{s}^{b,n+1}_{t}, & 1 \leq n < N_{b,t}.\\
        \boldsymbol{s}^{b,1}_{t+1}, & n = N_{b,t}.
        \end{cases}
\end{equation}
where $\boldsymbol{s}^{b,1}_{t+1}$ represents the system state of the first task in the time slot $t+1$.

\textbf{Action space.} We use a binary vector $\boldsymbol{a}^{b,n}_{t}$ = $[a^{b,n}_{t,1}$, $a^{b,n}_{t,2}$, $\cdots$, $a^{b,n}_{t,B}]$ to represent the action of a task $n$ that is arrived to the BS $b$ at time slot $t$. Here, $a^{b,n}_{t,b'}$ equals 1 if the ES $b'$ is selected for the task offloading and 0 otherwise. Obviously, the $\boldsymbol{a}^{b,n}_{t}$ has $2^{B}-1$ different values since at least one ES is selected during task offloading. Let $\mathcal{A}$ represent the set of all action spaces. We have $\mathcal{A}$ = $\{[1, 0, 0, \cdots, 0] $, $[0, 1, 0, \cdots, 0]$, $\cdots$, $[1, 1, 1, \cdots, 1]\}$. Furthermore, $\boldsymbol{a}^{b,n}_{t}$ is formulated as
\begin{equation}
	\boldsymbol{a}^{b,n}_{t} = \begin{cases} \mathop{\arg\max}\limits_{\boldsymbol{a}\in\mathcal{A}}Q(\boldsymbol{s}^{b,n}_{t},\boldsymbol{a};\boldsymbol{\theta}_{b}),\quad  & \text{w.p.}\; \epsilon,\\
                 \text{A random action from } \mathcal{A},\quad & \text{w.p.}\; 1 - \epsilon, \end{cases}
\end{equation}
where w.p. represents ``with probability''. $\epsilon$ is the greedy probability of ES (i.e., action) selection in our model. Intuitively, with probability $\epsilon$, the scheduler chooses the action that leads to the maximal Q-value under the inputting state. Otherwise, $\boldsymbol{a}^{b,n}_{t}$ is set a random action from $\mathcal{A}$ with probability $1 - \epsilon$.

\textbf{Reward.} Note that the offloading make-span of a task $n$ arrived at BS $b$ at time slot $t$ equals $T^{\text{make}}_{b,n,t}$. Our objective is to minimize the average total offloading make-span across all tasks over all time slots. We hope the less the make-span, the more the reward. Thus, the negative task offloading make-span is defined as the reward in our model. Let a variable $r^{b,n}_{t}$ present the reward of the task $n$ under giving the state $\boldsymbol{s}^{b,n}_{t}$ and action $\boldsymbol{a}^{b,n}_{t}$. We have
\begin{equation}\label{Eq11}
    r^{b,n}_{t}(\boldsymbol{s}^{b,n}_{t},\boldsymbol{a}^{b,n}_{t}) = \begin{cases} -T^{\text{make}}_{b,n,t}, & T^{\text{make}}_{b,n,t} < \tau_{n}.\\
	-10\cdot\tau_{n}, & \mbox{Otherwise}. \end{cases}
\end{equation}
Here, $-10 \cdot \tau_{n}$ is the penalty when the task is not completed before the deadline $\tau_{n}$. Moreover, the penalty value (i.e., $-10 \cdot \tau_{n}$) can be defined by the user, as long as it is significantly lower than the negative value of the maximum task service delay (i.e., deadline $\tau_{n}$), regardless of the specific number. Obviously, the $-10 \cdot \tau_{n}$ is much lower than the $-\tau_{n}$.

\subsection{Algorithm Design for Workload Allocation}\label{worload-allocation-section}
For each task $n\in\mathcal{N}_{b,t}$, when $a^{b,n}_{t,b'} = 0$, the $x_{b,n,t,b'}$ should equals 0. Therefore, we just need to solve the $x_{b,n,t,b'}$ when $a^{b,n}_{t,b'} \neq 0$. Let $\mathcal{L}^{b,n}_t$ denote the index set of selected ESs that contains each $b'\in \mathcal{B}$ with $a^{b,n}_{t,b'} \neq 0$. Then, since the action $\boldsymbol{a}^{b,n}_{t}$ of ES selection has been solved in the first stage, the workload allocation problem can be formulated as
\begin{align}\label{workload-allocation-p}
	&\min_{\boldsymbol{x}_{b,n,t}} \max \{T_{b,n,t,b'}^{\text{proc}}\}_{b' \in \mathcal{L}^{b,n}_t}  \\
	&\mbox{s.t.}\quad \sum\nolimits_{b'\in \mathcal{L}^{b,n}_t}x_{b,n,t,b'}=1.\nonumber
\end{align}

Furthermore, to achieve the minimal solution for the problem (\ref{workload-allocation-p}), we propose a CWA algorithm to achieve the optimal $\boldsymbol{x}_{b,n,t}$. Specifically, we observe that when the $\boldsymbol{a}^{b,n}_{t}$ is solved in the first stage, the subworkloads processing delays $\{T_{b,n,t,b'}^{\text{proc}}\}_{b' \in \mathcal{L}^{b,n}_t}$ of task $n$ and the number $k$ (i.e., $|\mathcal{L}^{b,n}_t|$) of selected ESs are all identified. The minimal offloading make-span of a task exists when each value $T_{b,n,t,b'}^{\text{proc}}$ for $b' \in \mathcal{L}^{b,n}_t$ is equal. Then, based on the information obtained in the ES selection stage, the optimal allocation fraction $\boldsymbol{x}_{b,n,t}$ of the task workload can be determined in the following three cases.

\begin{itemize}
\item[$\bullet$] \textbf{Case 1 ($k=1$):} In this case, only one ES is selected to process a task by the AEQN model. Assuming the ES $b'$ is selected for a task $n$. Then, the entire workload $d_{n}$ is allocated to the ES $b$ for processing. Thus, we have $x_{b,n,t,b'}$ = 1 and other elements in $\boldsymbol{x}_{b,n,t}$ equal to 0.

\item[$\bullet$] \textbf{Case 2} ($k=2$): In this case, two different ESs $i$ and $j$ (i.e., $\mathcal{L}^{b,n}_t$ = $\{i, j\}$) are selected to collaboratively process the workload $d_{n}$ by the AEQN model. At this time, the problem (\ref{workload-allocation-p}) is expressed to 
\begin{alignat}{2}\label{Eq14}
	&\min_{\boldsymbol{x}_{b,n,t}} \max \{T_{b,n,t,i}^{\text{proc}},T_{b,n,t,j}^{\text{proc}}\} \\
	&\mbox{s.t.}\quad x_{b,n,t,i} + x_{b,n,t,j} = 1. \nonumber
\end{alignat}
Obviously, the optimal allocation fraction is existed when $T_{b,n,t,i}^{\text{proc}} = T_{b,n,t,j}^{\text{proc}}$. Then, the problem (\ref{Eq14}) is equal to solve the variables $x_{b,n,t,i}$ and $x_{b,n,t,j}$ by bellowing equations
\begin{equation}
	\begin{cases} T_{b,n,t,i}^{\text{sum}} + T^{\text{wait}}_{b,n,t,i} = T_{b,n,t,j}^{\text{sum}} + T^{\text{wait}}_{b,n,t,j}\\
	x_{b,n,t,i} + x_{b,n,t,j} = 1, \end{cases}
\end{equation}
where  $T_{b,n,t,i}^{\text{sum}} = x_{b,n,t,i}\cdot \left(\frac{d_{n}}{v_{b,i,t}} + \frac{ \rho_{n} \cdot d_{n}}{f_{i}}\right)$ and $T_{b,n,t,j}^{\text{sum}} = x_{b,n,t,j}\cdot \left(\frac{d_{n}}{v_{b,j,t}} + \frac{ \rho_{n} \cdot d_{n}}{f_{j}}\right)$. Afterward, by solving the above two equations (e.g., using Gaussian elimination), the $x_{b,n,t,i}$ and $x_{b,n,t,j}$ can be achieved as 
\begin{align}\label{Eq-x-b-solution}
    x_{b,n,t,b'} = \frac{T^{\text{sum}}_{b,n,t,i}+T^{\text{wait}}_{b,n,t,i}-T^{\text{wait}}_{b,n,t,b'}}{T^{\text{sum}}_{b,n,t,b'}+T^{\text{sum}}_{b,n,t,i}},
\end{align}
and 
\begin{align}\label{Eq-x-i-solution}
    x_{b,n,t,i} = \frac{T^{\text{sum}}_{b,n,t,b'}+T^{\text{wait}}_{b,n,t,b'}-T^{\text{wait}}_{b,n,t,i}}{T^{\text{sum}}_{b,n,t,b'}+T^{\text{sum}}_{b,n,t,i}},
\end{align}
respectively.
\item[$\bullet$] \textbf{Case 3} ($3 \leq k\leq B$): In this case, 3, 4, $\cdots$, or $B$ different ESs are respectively selected to collaboratively process the workload $d_{n}$. Similarly, the optimal allocation fraction exists when subworkload processing delays $\{T_{b,n,t,b'}^{\text{proc}}\}_{b' \in \mathcal{L}^{b,n}_t}$ are equal at the selected 3, 4, $\cdots$, or $B$ ESs. Then, we can formulate 3, 4, $\cdots$, or $B$ equations when 3, 4, $\cdots$, or $B$ ESs are respectively selected. Finally, for $k\in\{3,\cdots,B\}$, the $\boldsymbol{x}_{b,n,t}$ can be achieved by solving these equations with the existing method (e.g., Gaussian elimination), respectively.
\end{itemize}

Based on the above cases, the CWA algorithm is implemented as shown in Algorithm \ref{CWA-procedure}. Algorithm \ref{CWA-procedure} works as follows. Firstly, the fraction $\boldsymbol{x}_{b,n,t}$ of workload allocation is initialized to 0 (Line 1). The index set $\mathcal{L}^{b,n}_t$ and the number $k$ of selected ES are achieved according to the input action of ES selection (Lines 2 and 3). Secondly, if $k=1$, we can get the index $b'$ of the selected ES that is contained in $\mathcal{L}^{b,n}_t$. Furthermore, the workload fraction $x_{b,n,t,b'}$ allocated to the selected ES $b'$ is set to 1. If $k > 1$, the $k$ equations are built according to the previous \textbf{Cases 2} and \textbf{3} (Lines 8-11). Furthermore, the workload allocation fraction $\boldsymbol{x}_{b,n,t}$ can be achieved by solving these equations with the existing method $linalg.solve()$ in the open-source Numpy library (Line 12).

\begin{algorithm}[!t]
\DontPrintSemicolon
  \KwInput{The action $\boldsymbol{a}^{b,n}_{t}$ of ES selection;}
  \KwOutput{The fraction $\boldsymbol{x}_{b,n,t}$ of workload allocation;}
    Initialize variable $\boldsymbol{x}_{b,n,t}$ = 0;\\
    Achieve the $\mathcal{L}^{b,n}_t=\{b'\in\mathcal{B}:a^{b,n}_{t,b'}=1\}$;\\   
    Calculate the $k = |\mathcal{L}^{b,n}_t|$;\\
    \eIf{$k = 1$}{
        $b'\in\mathcal{L}^{b,n}_t$;\\
        $x_{b,n,t,b'} = 1$;\\
    }
    {
        \For{\rm{each $b' \in \mathcal{L}^{b,n}_t$}}
            {
                Calculate the waiting time $T^{\text{wait}}_{b,n,t,b'}$ by Eqn. (\ref{waiting-time});\\
                Calculate the variable coefficient $T_{b,n,t,b'}$ of ES decision, i.e., $T_{b,n,t,b'}$ = $\left(\frac{d_{n}}{v_{b,b',t}} + \frac{ \rho_{n} \cdot d_{n}}{f_{b'}}\right)$;\\
            }
        According to the \textbf{Cases 2} and \textbf{3}, build the corresponding coefficients of $k$ equations based on the achieved $\boldsymbol{T}^{\text{wait}}_{b,n,t}$ = $\{T^{\text{wait}}_{b,n,t,b'}\}_{b' \in \mathcal{B}}$, $\boldsymbol{T}_{b,n,t}$ = $\{T_{b,n,t,b'}\}_{b' \in \mathcal{B}}$, and Eqn. (\ref{Eq-constrain});\\
        Achieve the workload allocation fraction $\boldsymbol{x}_{b,n,t}$ = $linalg.solve$ (the coefficients of $k$ equations);\\
    }
\caption{CWA algorithm}
% \vspace{-0.25cm}
\label{CWA-procedure}
\end{algorithm}

\begin{theorem}\label{CWA-complexity-Theorem}
The time complexity of Algorithm \ref{CWA-procedure} is $O(k^3)$.
\end{theorem}

\begin{proof}
    See Appendix D \cite{huang2022throughput}.
\end{proof}

\subsection{Improved Algorithm for Workload Allocation}\label{improved-algorithm}
Compared with our previous work \cite{xu2024enhancing}, this section presents an improved workload-allocation algorithm, along with a corresponding theoretical performance analysis. As shown in \textbf{Theorem \ref{CWA-complexity-Theorem}}, the proposed CWA algorithm has the high time complexity $O(k^{3})$, which affects the running efficiency of our AMCoEdge method. To address this issue, we design an HECWA algorithm to quickly compute workload allocation fractions. In particular, the task's queuing time is deficient, especially for small task workloads and large ESs with high computing capacity. That is, the task's queuing time can be ignored in our system. Thus, the Eqn. (\ref{Eq-process-delay}) is rephrased as 
\begin{equation*}
    T_{b,n,t,b'}^{\text{proc}} = x_{b,n,t,b'} \cdot\left(\frac{d_{n}}{v_{b,b',t}} + \frac{ \rho_{n} \cdot d_{n}}{f_{b'}}\right).
\end{equation*}
At this time, the workload allocation fraction $x_{b,n,t,b'}$ for $b' \in \mathcal{L}^{b,n}_{t}$ can be optimally achieved by the Eqn. (\ref{closed-form-formula}), as shown in the following \textbf{Theorem \ref{effectiveness-Theorem}}. 

\begin{theorem}\label{effectiveness-Theorem}
Given the index set $\mathcal{L}^{b,n}_{t}$ of selected ESs and $\forall\; b'\in \mathcal{L}^{b,n}_{t}$, when the queuing time can be ignored, the optimal workload allocation fraction $x_{b,n,t,b'}$ of our problem (\ref{workload-allocation-p}) is achieved by the following formula:
\begin{equation}\label{closed-form-formula}
    x_{b,n,t,b'} = \begin{cases} \quad \quad \quad \quad 1, & k=1;\\
    \frac{\prod_{u\in\mathcal{L}^{b,n}_{t}\setminus b'} T_{b,n,t,u}}{\sum_{v\in\mathcal{L}^{b,n}_{t}}\prod_{u\in\mathcal{L}^{b,n}_{t}\setminus v} T_{b,n,t,u}}, & k > 1, \end{cases}
\end{equation}
where $\prod$ means continuous multiplication, $u\in\mathcal{L}^{b,n}_{t}\setminus b'$ and $u\in\mathcal{L}^{b,n}_{t}\setminus v$ represent that $u$ belongs to the $\mathcal{L}^{b,n}_{t}$ except the $b'$ and $v$, respectively, and $T_{b,n,t,u} = \frac{d_{n}}{v_{b,u,t}} + \frac{ \rho_{n} \cdot d_{n}}{f_{u}}$.
\end{theorem}

\begin{proof}
See Appendix E.
\end{proof}

Based on the above analyses, the HECWA algorithm is presented in Algorithm \ref{HECWA-algorithm}. The Algorithm \ref{HECWA-algorithm} begins by initializing the variables $T^{\text{sum}}$ (i.e., the sum of the task processing delays at the selected $k$ ESs) and $\boldsymbol{x}_{b,n,t}$ to 0 (Line 1). Then, similar to the Algorithm \ref{CWA-procedure}, the $\mathcal{L}^{b,n}_t$ and $k$ are achieved in lines 2 and 3, respectively. Furthermore, if $k = 1$, the allocation fraction $x_{b,n,t,b'}$ is achieved by the lines 4-6. However, when the $k > 1$, for each $b' \in \mathcal{L}^{b,n}_{t}$, the $x_{b,n,t,b'}$ is calculated according to the Eqn. (\ref{closed-form-formula}) (Lines 7-11).

\begin{theorem}\label{HECWA-complexity-Theorem}
The time complexity of Algorithm \ref{HECWA-algorithm} is $O(k)$.
\end{theorem}

\begin{proof}
See Appendix F.
\end{proof}
\begin{algorithm}[!t]
\DontPrintSemicolon
  \KwInput{The action $\boldsymbol{a}^{b,n}_{t}$ of ES selection;}
  \KwOutput{The fraction $\boldsymbol{x}_{b,n,t}$ of workload allocation;}
    Initialize variables $T^{\text{sum}} = 0$, $\boldsymbol{x}_{b,n,t}$ = 0;\\
    Achieve the $\mathcal{L}^{b,n}_t=\{b'\in\mathcal{B}:a^{b,n}_{t,b'}=1\}$;\\
    Calculate the $k = |\mathcal{L}^{b,n}_{t}|$;\\
    \eIf{$k = 1$}{
        $b' \in \mathcal{L}^{b,n}_{t}$;\\
        $x_{b,n,t,b'} = 1$;\\
    }
    {
        \For{\rm{each $b' \in \mathcal{L}^{b,n}_{t}$}}
             {
             $T^{\text{sum}}$ =  $T^{\text{sum}}$ +  $\prod_{u\in\mathcal{L}^{b,n}_{t}\setminus b'} T_{b,n,t,u}$;\\
             }
        \For{\rm{each $b' \in \mathcal{L}^{b,n}_{t}$}}
            {
            $x_{b,n,t,b'} = \frac{\prod_{u\in\mathcal{L}^{b,n}_{t}\setminus b'} T_{b,n,t,u}}{T^{\text{sum}}}$;\\
            }
    }
\caption{HECWA algorithm}
\label{HECWA-algorithm}
\end{algorithm}

\subsection{Overall Algorithm Implementation}\label{sec4-4}
We present the proposed AMCoEdge method in Algorithm \ref{AMCoEdge-algorithm} to get a full picture of its workflow. The Algorithm \ref{AMCoEdge-algorithm} works as follows. 

1) The variables $actStep$ and $trainStep$ of the action and training steps are initialized to 0, respectively (Line 1). The AEQN and ATQN are initialized with random parameters $\boldsymbol{\theta}_{b}$ and $\boldsymbol{\theta}'_{b}$, respectively (Lines 2)

\begin{algorithm}[!t]
\DontPrintSemicolon
\KwInput{The task data $\{d_{n}\}_{n\in\mathcal{N}_{b,t}, b\in\mathcal{B}, t\in\mathcal{T}}$;}
\KwOutput{The allocation fraction of task workload $\{x_{b,n,t,b'}\}_{b,b'\in\mathcal{B}, n\in\mathcal{N}_{b,t}, t\in\mathcal{T}}$;}
Initialize variables $actStep = 0$, $trainStep = 0$;\\
Initialize the AEQN and ATQN of each BS $b\in \mathcal{B}$ with random $\boldsymbol{\theta}_{b}$ and $\boldsymbol{\theta}'_{b}$, respectively;\\
  \For{\rm{episode = 1, 2, ..., $E$}}
    {Initialize system environment;\\
      \ForEach{\rm{time slot $t\in \mathcal{T}$}}
        {   \For{\rm{all BS $b \in \mathcal{B}$ in parallel}}
            {
                \ForEach{\rm{new arrival task $n \in \mathcal{N}_{b,t}$}}
                {                  
                    Observe the $\boldsymbol{s}^{b,n}_{t}$;\\
                    Generate the $\boldsymbol{a}^{b,n}_{t}$ by the AEQN model;\\
                    Based on $\boldsymbol{a}^{b,n}_{t}$, solve $\boldsymbol{x}_{b,n,t}$ by the proposed CWA or HECWA algorithm;\\
                    Allocate the $\{d_{n}\cdot \boldsymbol{x}_{b,n,t}\}$ to the selected ESs for parallel processing and calculate the $r^{b,n}_{t}$ by Eqn. (\ref{Eq11});\\
                    Observe the $\boldsymbol{s}^{b,n}_{t,\text{next}}$ and store the $(\boldsymbol{s}^{b,n}_{t}$, $\boldsymbol{a}^{b,n}_{t}$, $r^{b,n}_{t}$, $\boldsymbol{s}^{b,n}_{t,\text{next}})$ to experience pool;\\ 
                }       
                       
                $actStep$++;\\
                \If{\rm{$actStep$ $>$ 200 and $actStep$ \% 10 == 0}}
                {     
                    Extract a set of samples $\mathcal{I}$ from the experience pool and train the AEQN model to update its $\boldsymbol{\theta}_{b}$ by minimizing $L(\boldsymbol{\theta}_{b})$ in Eqn. (\ref{loss-function});\\
                    $trainStep$++;\\
                    \If{\rm{$trainStep$ \% $\eta$ == 0}}
                    {
                        Update the ATQN parameter $\boldsymbol{\theta}'_{b}$ by copying $\boldsymbol{\theta}_{b}$;\\ 
                    } 
                }     
            }
            Update $\boldsymbol{q}_{t}$ by the Eqn. (\ref{queue-update});\\
        }     
    }
\caption{Online distributed AMCoEdge algorithm}
\label{AMCoEdge-algorithm}
\end{algorithm}

2) For each episode = $1, 2, \cdots, E$, the algorithm begins to initialize the system environment (Lines 3 and 4).

3) For each time slot $t \in \mathcal{T}$ and task $n \in \mathcal{N}_{b,t}$ in parallel (Lines 5 - 7), the current state $\boldsymbol{s}^{b,n}_{t}$ is observed from the system environment (Line 8). Then, by utilizing the $\boldsymbol{s}^{b,n}_{t}$ as input, the AEQN generates the ES selection action $\boldsymbol{a}^{b,n}_{t}$ (Line 9). Accordingly, based on the achieved $\boldsymbol{a}^{b,n}_{t}$, the workload allocation decision $\boldsymbol{x}_{b,n,t}$ is solved by the proposed CWA or HECWA algorithm (Line 10). Furthermore, the workload $\boldsymbol{x}_{b,n,t} \cdot \rho_{n} \cdot d_{n}$ of the task $n$ is allocated to the selected ESs $\mathcal{L}_{t}^{b,n}$ for collaborative processing. Meanwhile, the reward $r^{b,n}_{t}$ is calculated using the Eqn. (\ref{Eq11}) (Line 11). The next state $\boldsymbol{s}^{b,n}_{t,\text{next}}$ is observed from the system environment and the transition tuple $(\boldsymbol{s}^{b,n}_{t}, \boldsymbol{a}^{b,n}_{t}, r^{b,n}_{t}, \boldsymbol{s}^{b,n}_{t,\text{next}})$ is stored to the experience pool (Line 12). Afterward, the variable $actStep$, which tracks the action steps, is incremented by 1 (Line 13). 

4) All ESs will train their AEQN models in parallel. Specifically, for all BS $b\in\mathcal{B}$ in parallel, if the $actStep$ exceeds 200 (i.e., after the AEQN has executed more than 200 times) and is a multiple of 10 (i.e., after every 10 executions) (Line 14), a bath sample from the experience pool is extracted to train the AEQN model to achieve its model parameter $\boldsymbol{\theta}_{b}$ (Line 15). Here, the values 200 and 10 are set to ensure effective model training based on user-defined real-world conditions \cite{tang2022deep}. In this process, the key idea of AEQN model training is to minimize the difference between the AEQN Q-values and the target Q-values under ATQN, based on these experience samples. Let $\mathcal{I}$ denote the set of samples and the $i^{th}$ sample is $(\boldsymbol{s}^{b,n}_{t}(i)$, $\boldsymbol{a}^{b,n}_{t}(i)$, $r^{b,n}_{t}(i)$, $\boldsymbol{s}^{b,n}_{t,\text{next}}(i))$. Then, based on these experience samples in $\mathcal{I}$, the parameter $\boldsymbol{\theta}_b$ is updated by minimizing the following loss function $L(\boldsymbol{\theta}_{b})$.
\begin{align}\label{loss-function}
    L(\boldsymbol{\theta}_{b}) &= \frac{1}{|\mathcal{I}|}\sum_{i\in\mathcal{I}}\left(Q_{b}(\boldsymbol{s}^{b,n}_{t}(i),\boldsymbol{a}^{b,n}_{t}(i);\boldsymbol{\theta}_{b}) - Q^{\text{target}}_{b,i}\right)^{2},
\end{align}
where $Q^{\text{target}}_{b,i}$ is the target Q-value under ATQN for a sample $i \in \mathcal{I}$ and is derived as
\begin{align}\label{target-Q-value}
    Q^{\text{target}}_{b,i} = r^{b,n}_{t}(i) + \gamma\cdot \mathop{\max}\limits_{\boldsymbol{a}\in\mathcal{A}} Q(\boldsymbol{s}^{b,n}_{t,\text{next}}(i),\boldsymbol{a};\boldsymbol{\theta}'_{b}),
\end{align}
where $\gamma$ is a reward decay factor. 

5) After the above fourth step, the variable $trainStep$, which records the training steps, is increased by 1 (Line 16). Then, if the $trainStep$ is a multiple of $\eta$ (i.e., after every $\eta$ training sessions) (Line 17), the ATQN parameter $\boldsymbol{\theta}'_{b}$ is updated by copying the AEQN parameter $\boldsymbol{\theta}_{b}$ (Line 18). Here, $\eta$ is denoted as the number of training sessions after which ATQN needs to be updated. That is, for every $\eta$ training session, the ATQN parameter $\boldsymbol{\theta}'_{b}$ has to be updated by copying the AEQN parameter $\boldsymbol{\theta}_{b}$. This step aims to keep the ATQN parameter $\boldsymbol{\theta}'_{b}$ up-to-date. As a result, the computation of the target Q-values in Eqn. (\ref{target-Q-value}) can better approximate the long-term minimal make-span in our model.

6) At the end of each time slot, the queue $\boldsymbol{q}_{t}$ is updated according to the Eqn. (\ref{queue-update}). Through the above steps, the near-optimal workload allocation policies $\boldsymbol{x}$ are gradually learned and refined.

\begin{theorem}\label{AMCoEdge-complexity-Theorem}
For each time slot $t$, the time complexity of our Algorithm \ref{AMCoEdge-algorithm} is approximately equal to the linear time complexity $O(N_{b,t})$.
\end{theorem}

\begin{proof}
See Appendix G.
\end{proof}

\begin{rem}
    Based on the previous Subsection \ref{improved-algorithm} and \textbf{Theorem} \ref{HECWA-complexity-Theorem}, our Algorithm \ref{HECWA-algorithm} will quickly achieve the optimal solution of our problem (\ref{workload-allocation-p}) since $1 \leq k \leq B$ and $k$ is usually a small value. Thus, the execution time of workload allocation can be ignored. Furthermore, by combining with Subsection \ref{sec4-4} and \textbf{Theorem} \ref{AMCoEdge-complexity-Theorem}, our Algorithm \ref{AMCoEdge-algorithm} will achieve the near-optimal solution of our AM-CMEC problem within an approximate linear time complexity.
\end{rem}

\section{Performance Evaluation}\label{evaluation}
This section presents the evaluation setting, baselines, and results with insightful analysis. Our simulations run on a Windows 10 computer with 32GB of RAM and an Intel Core i7 2.2 GHz processor, using PyCharm 2023. We have released our source code at \url{https://github.com/ChangfuXu/AMCoEdge}.

\subsection{Experimental Setting}
We consider that the collaborative MEC system has several BSs with different ES computing capacities in our experiments. We use the TensorFlow machine learning framework to implement our ADQN model. For simplicity and convenience, the AEQN and ATQN in our method are both implemented as fully connected neural networks with one input layer, two hidden layers with 20 neurons, and one output layer. The number of episodes $E$ is set to 300 according to the results as shown in Fig. \ref{Fig4}. The greedy probability $\epsilon$ is gradually increasing from 0 to 0.99 with a step of 0.001. The reward decay factor $\gamma$ is 0.9. The storage size of the experience pool at each ES is set to 500. The parameters learning rate, batch size, updating interval $\eta$, and optimizer are set to 0.001, 32, 100, and Adam, respectively. 

% \begin{table}[!t]
% \caption{Default environment parameters in our experiment.}
% \centering
% \begin{tabular}{p{16pt}<{\centering}p{86pt}<{\centering}p{16pt}<{\centering}p{65pt}<{\centering}}
%     \toprule
%     Param. & Value & Param. & Value\\
%     \midrule
%     $B$ & 5 \cite{fan2024collaborative, xu2024dynamic} & $|\mathcal{T}|$ & 60 \cite{chu2023online}\\
%     $N_{b,t}$ & 50 \cite{chu2023online} &  $d_{n}$ & [2,5] Mbits \cite{chu2023online} \\
%     $\Delta$ & 1.0 second \cite{tang2022deep} & $f_{b'}$ & [10,50] GHz \cite{fan2024collaborative} \\ 
%     $\rho_{n}$ & [100,300] cycles/bit \cite{fan2024collaborative} &  $\tau_{n}$ & 1.0 second \cite{tang2022deep} \\ 
%     $v_{b,b',t}$ & [400,500] Mbits/s \cite{fan2024collaborative} &  $p_{n}$ & 0.3 \cite{tang2022deep}\\   
%     \bottomrule
% \end{tabular}
% \label{table3}
% \end{table}

According to \cite{xu2024dynamic, fan2024collaborative, tang2022deep, chu2023online}, the related environment and training parameters in our experiment are set as follows by default unless specified. Our experiments consider a collaborative MEC network with 5 (i.e., $B$ = 5) BSs. The number of tasks $N_{b,t}$ arriving at each BS is set to 50. The task size $d_{n}$ and computation density $\rho_{n}$ are randomly generated in the range [2, 5] Mbits and [100, 300] CPU cycles/bit, respectively. The computing capacity $f_{b'}$ and transmission rate $v_{b,b',t}$ are set to uniform distributions over [10, 50] GHz and [400, 500] Mbits/s, respectively. The total number $|\mathcal{T}|$ of time slots is 60 (i.e., 1 minute). The length $\Delta$ of each time slot $t$ and the processing deadline $\tau_n$ are all set to 1.0 seconds. To simulate the task arrival trace in the real world, the tasks are generated from UDs with probability $p_{n} = 0.3$, i.e., $d_{n} = d_{n}$ w.p. $p_{n}$ and 0 otherwise.

\subsection{Baselines}
We use four baselines to compare our method's performance as follows.
\begin{itemize}
    \item \textbf{RandCoEdge:} A random-based method that randomly selects an ES to process each task with the local ES. For fairness, the RandCoEdge decision $\boldsymbol{x}$ is satisfied with the same constraints as our method in the experiment.
    \item \textbf{DRLCoEdge \cite{li2020deep}:} A representative DRL-based method that selects an ES to collaboratively process each task with the local ES based on DRL.
    \item \textbf{SMCoEdge \cite{xu2023smcoedge,xu2024dynamic}:} A state-of-the-art method that simultaneously selects suitable $k$ ESs to collaboratively process each offloading task based on top-$k$ DQN model. However, although the $k$ can be previously adjusted by the user, it is fixed in execution. In our experiment, we keep the same setting $k=3$ in \cite{xu2023smcoedge,xu2024dynamic}.
    \item \textbf{Optimal:} An optimal method that greedily selects the most suitable ESs to collaboratively process each task by enumerating all action spaces. However, this method is infeasible, since the scheduler cannot know the available computing and networking resources of ESs in a practical MEC system.
\end{itemize}

\subsection{Experimental Results}
This subsection first evaluates the convergence of our method and baselines by varying the number of episodes, as shown in Fig. \ref{Fig4}, where the x-axis shows the episode number and the y-axis shows the average make-span across tasks and time slots in each episode. Then, we compare our method and baselines on task offloading, make-span, and failure rate (i.e., the ratio of dropped tasks to total arrival tasks), as shown in Figs. \ref{Fig5} and \ref{Fig6}, respectively.

\subsubsection{Convergence Analysis of Different Episodes}
The results of Fig. \ref{Fig4} are achieved by varying the number of episodes from 1 to 1000. From Fig. \ref{Fig4}, the average make-spans of DRLCoEdge, SMCoEdge, and AMCoEdge methods initially drop erratically and then stabilize with the increase of episodes from 1 to 1000. The average make-spans of the RandCoEdge and Optimal methods stabilize as the episode increases. However, our AMCoEdge method achieves the lowest make-span and outperforms RandCoEdge, DRLCoEdge, and SMCoEdge by about \textbf{28\%}, \textbf{17\%}, and \textbf{12\%}, respectively. Moreover, after 300 episodes, the average make-span of our method is close to that of the Optimal method. These results demonstrate that our AMCoEdge method effectively converges to a stable result and outperforms the RandCoEdge, DRLCoEdge, and SMCoEdge baselines in terms of average make-span.

\begin{figure}[!t]
    \centering
    \includegraphics[width=2.5in]{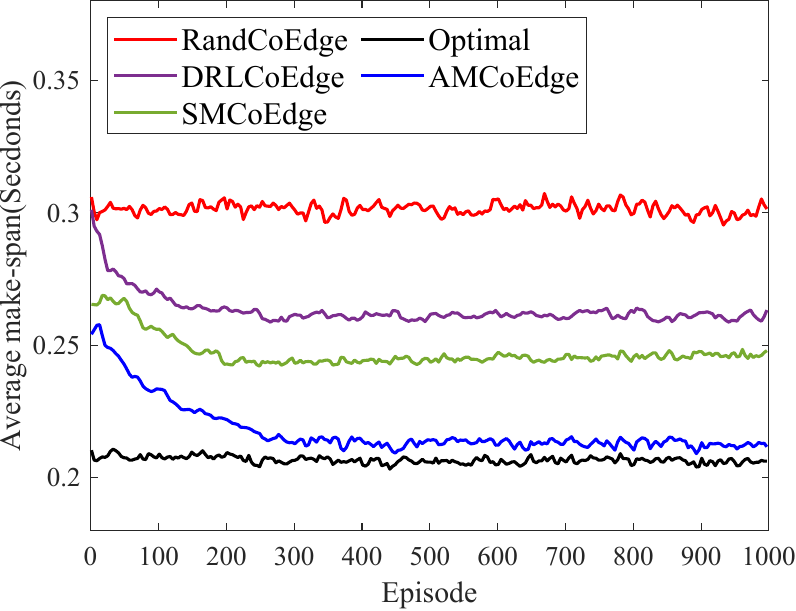}
    \caption{The make-span results of our method and baselines by varying the number of episodes.}
    \label{Fig4}
    % \vspace{-0.2cm}
\end{figure}

\subsubsection{Make-span Effects of Different Parameters}
In this subsection, we take a comparative experiment of our AMCoEdge method and baselines in terms of make-span by varying four parameters $N_{b,t}$, $p_{n}$, $f_{b'}$, and $\tau_{n}$, as given in Fig. \ref{Fig5}.

Especially, Fig. \ref{Fig5.sub.1} shows that the average make-spans of all five methods gradually increase as $N_{b,t}$ increases. However, the average make-span of our method is always lower than that of the RandCoEdge, DRLCoEdge, and SMCoEdge methods, and is slightly higher than that of the Optimal method. More precisely, when $N_{b,t}$ increases from 10 to 100, the average make-spans for the RandCoEdge, DRLCoEdge, and SMCoEdge methods increase from 0.0848 to 0.5358 seconds, 0.0721 to 0.4882 seconds, and 0.0634 to 0.4615 seconds, respectively. In contrast, the average make-span of our method rises only from 0.0514 to 0.4361 seconds, which outperforms that of the RandCoEdge, DRLCoEdge, and SMCoEdge methods by an average of \textbf{26.18\%}, \textbf{16.68\%}, and \textbf{11.04\%}, respectively. Meanwhile, the average make-span of the Optimal method increases from 0.0492 to 0.4125 seconds as $N_{b,t}$ increases from 10 to 100, and is only 5.83\% lower than that of our method on average. Additionally, when $N_{b,t}$ = 50 by default setting, the Optimal method just achieves a reduction of 0.0079 seconds in make-span. These behaviors are primarily attributed to AMCoEdge’s ability to adaptively select the near-optimal multiple ESs and simultaneously utilize their computing capacities to collaboratively process tasks, thereby reducing task offloading, make-span, and failure rate.

\begin{figure*}[!t]
	\centering 
	\subfigure[]{
		\label{Fig5.sub.1}
		\includegraphics[width=1.7in]{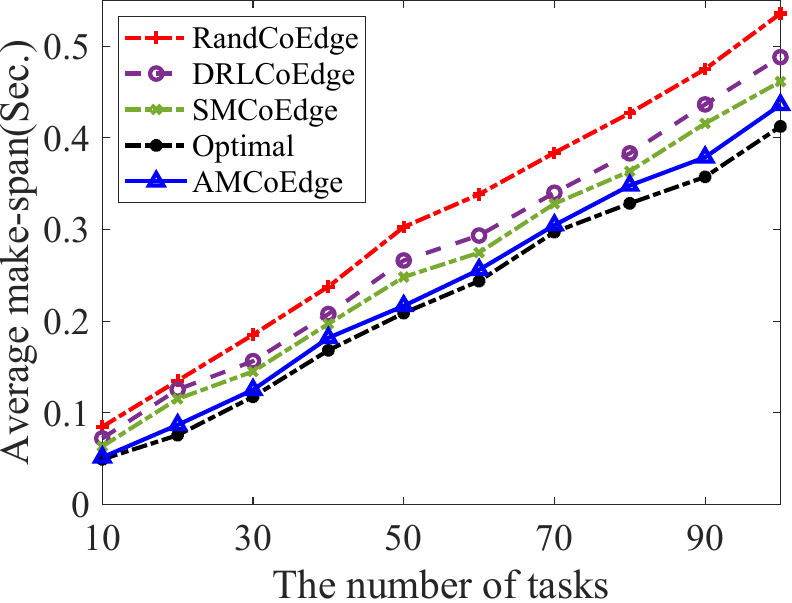}}
        \subfigure[]{
		\label{Fig5.sub.2}
		\includegraphics[width=1.7in]{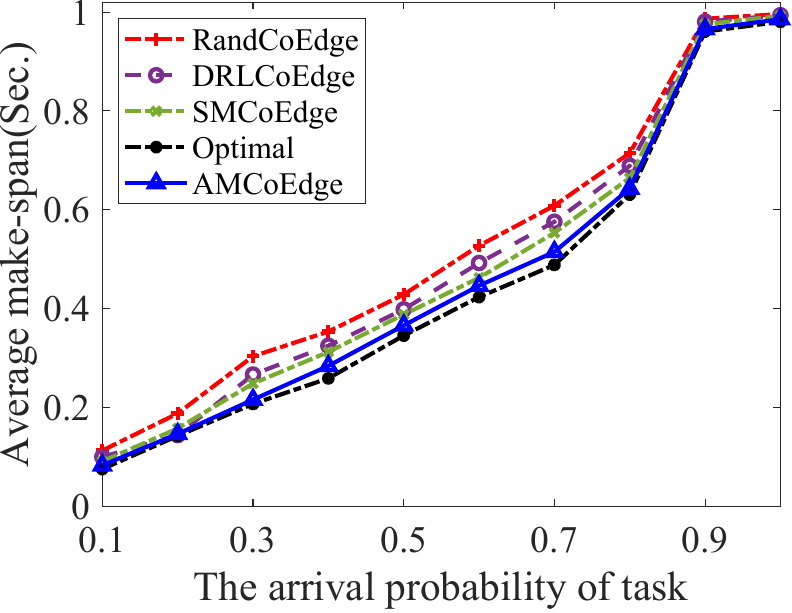}}
        \subfigure[]{
		\label{Fig5.sub.3}
		\includegraphics[width=1.7in]{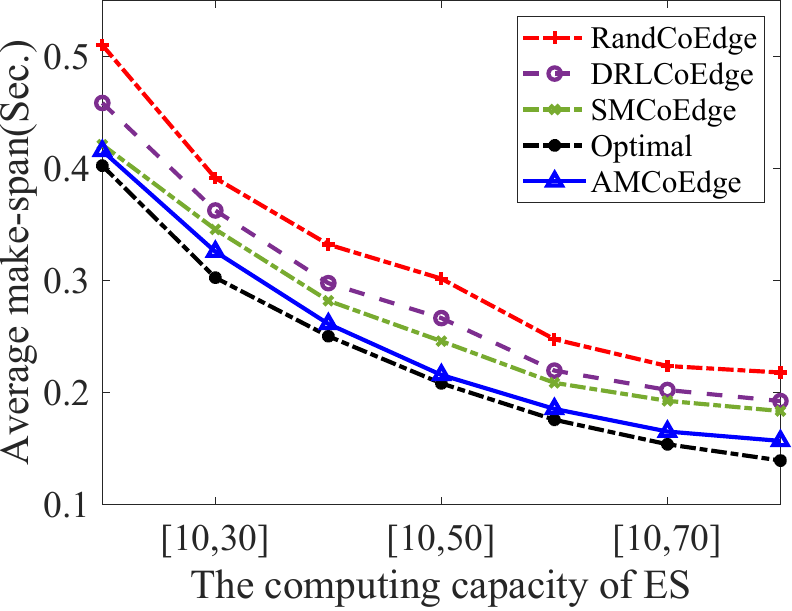}}
        \subfigure[]{
		\label{Fig5.sub.4}
		\includegraphics[width=1.7in]{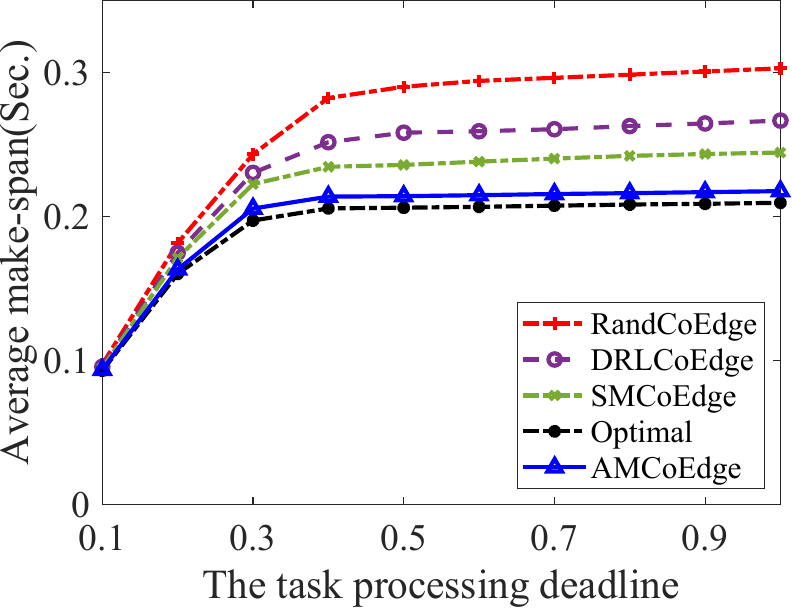}}
	\caption{Make-span comparisons of our method and baselines by varying four different parameters. (a) Varying the number $N_{b,t}$ of tasks. (b) Varying the arrival probability $p_{n}$ of tasks. (c) Varying the computing capacity $f_{b'}$ of ES. (d) Varying the task processing deadline $\tau_{n}$.}
	\label{Fig5}
\end{figure*}

\begin{figure*}[!t]
\centering 
\subfigure[]{
    \label{Fig6.sub.1}
    \includegraphics[width=1.7in]{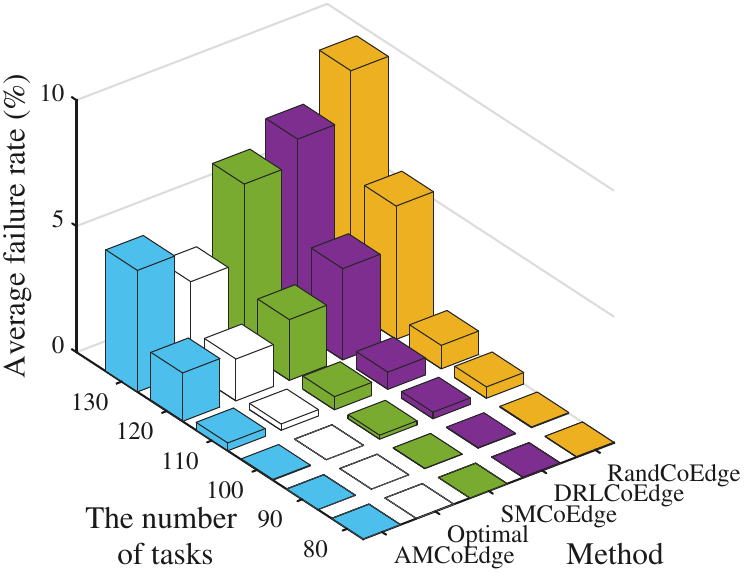}}
\subfigure[]{
    \label{Fig6.sub.2}
    \includegraphics[width=1.7in]{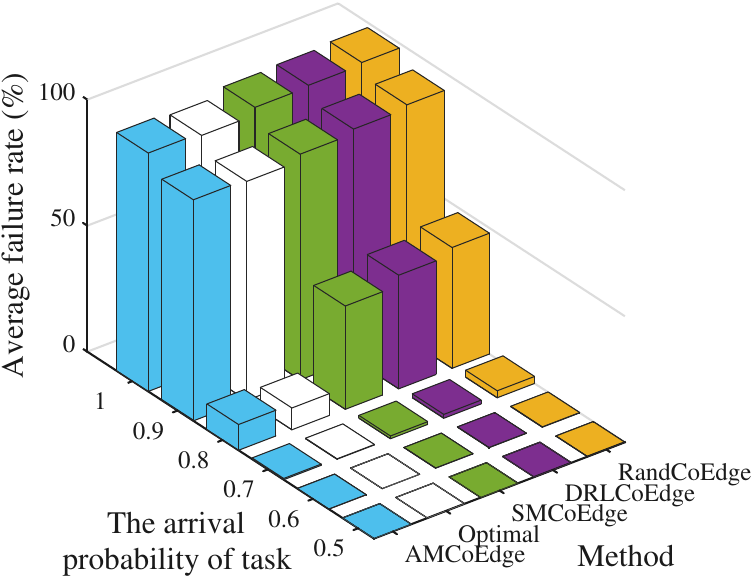}}
    \subfigure[]{
    \label{Fig6.sub.3}
    \includegraphics[width=1.7in]{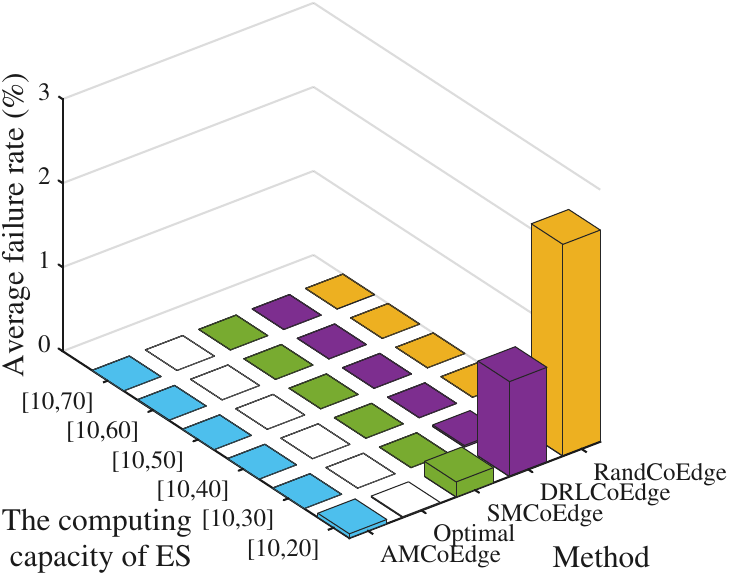}}
    \subfigure[]{
    \label{Fig6.sub.4}
    \includegraphics[width=1.7in]{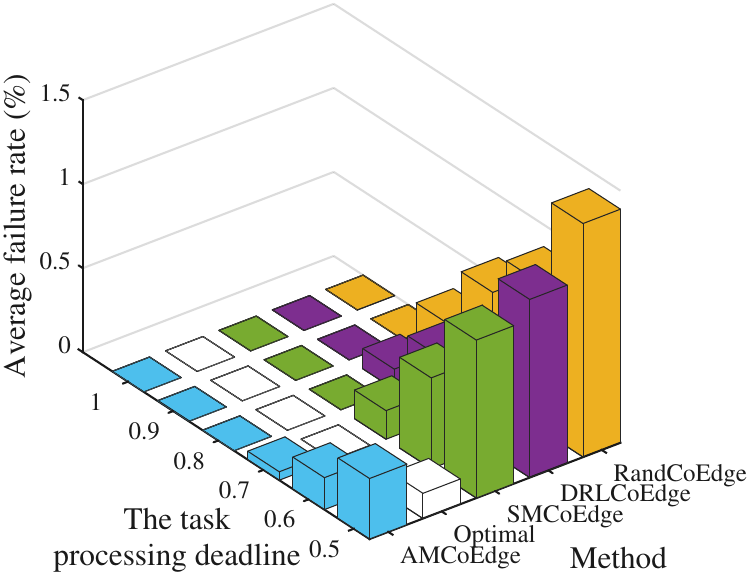}}
\caption{Failure rate comparisons of our method and baselines by varying four different parameters. (a) Varying the number $N_{b,t}$ of tasks. (b) Varying the arrival probability $p_{n}$ of tasks. (c) Varying the computing capacity $f_{b'}$ of ES. (d) Varying the task processing deadline $\tau_{n}$.}
\label{Fig6}
 % \vspace{-0.1cm}
\end{figure*}

Furthermore, we explore the average make-spans of all five methods as the arrival probability $p_{n}$ of tasks is varied, as shown in Fig. \ref{Fig5.sub.2}. From Fig. \ref{Fig5.sub.2}, we can see that the average make-spans of all five methods first increase gradually as $p_{n}$ increases from 0.1 to 0.8. After $p_{n}$ = 0.8, the average make-spans increase rapidly and tend towards stability until $p_{n}$ = 0.9. Apparently, our method consistently achieves the lowest make-span among RandCoEdge, DRLCoEdge, and SMCoEdge, and approximates the Optimal method's results as $p_{n}$ increases from 0.1 to 1.0. This is because as $p_{n}$ increases, the workloads at the ESs can increase. Then, when workloads exceed a certain threshold (e.g., all ESs' computing capacities), the make-span increases rapidly and stabilizes at the deadline.

Fig. \ref{Fig5.sub.3} indicates that as $f_{b'}$ increases from the range [10, 20] to [10, 80] GHz, the average make-spans of our method and four baselines all decrease gradually. However, our AMCoEdge method consistently exhibits the lowest make-span compared to the four baselines. Importantly, the make-span of our AMCoEdge method approximates the results of the Optimal method. Intuitively, with increased ES computing capacity, task processing speeds can be expedited, resulting in a lower make-span.

Also, we find that the average make-spans of our AMCoEdge method and four baselines are all first to decrease gradually as $\tau_{n}$ increases from 0.1 to 0.4, as shown in Fig. \ref{Fig5.sub.4}. Then, the average make-spans of them go to stabilization as $\tau_{n}$ increases from 0.4 to 1.0. However, our AMCoEdge method consistently yields the lowest make-span among the other three baselines, except for the Optimal method. Significantly, the average make-spans of our method closely approximate those of the Optimal method as $\tau_{n}$ increases. This behavior occurs because, with a larger deadline, nearly all tasks can be successfully processed, resulting in only a minor impact on the make-span when the deadline exceeds 0.4. Furthermore, the default setting of $\tau_{n}$ = 1.0 ensures a high offloading successful rate when the task is computation-intensive in the system.

\subsubsection{Failure Rate Effects of Different Parameters}
We also explore the effects of four parameters $N_{b,t}$, $p_{n}$, $f_{b'}$, and $\tau_{n}$ on the task offloading failure rate of our method and baselines. The experimental results are given in Fig. \ref{Fig6}.

In Fig. \ref{Fig6.sub.1}, it is evident that when the $N_{b,t}$ increases from 80 to 130, the offloading failure rates for RandCoEdge, DRLCoEdge, and SMCoEdge methods exhibit substantial increments from 0\% to 9.49\%, 0 to 7.59\%, and 0 to 6.62\%, respectively. In contrast, the offloading failure rates of our SMCoEdge method increase only from 0 to 4.82\%, representing the lowest failure rate compared to RandCoEdge, DRLCoEdge, and SMCoEdge methods and outperforming them by the average values of \textbf{78.82\%}, \textbf{54.38\%}, and \textbf{44.86\%}, respectively. Moreover, our method's offloading failure rates closely approximate the Optimal method's results (from 0 to 3.56\%) as $N_{b,t}$ increases from 80 to 130. Note that the failure rate of our method remains 0 across $N_{b,t}$ = 10-80, ensuring a low offloading rate in the default setting of $N_{b,t}$ = 50.

Fig. \ref{Fig6.sub.2} indicates that the average failure rates of all five methods increase gradually as the $p_{n}$ increases from 0.5 to 1. The average failure rate of our method is significantly lower than that of the RandCoEdge, DRLCoEdge, and SMCoEdge methods. The failure rates of our method always approximate those of the Optimal method as $p_{n}$ increases from 0.5 to 1.0, and only reach 0.48\% when $p_{n}$ equals 0.7.

Furthermore, Fig. \ref{Fig6.sub.3} reveals that the average failure rates of all five methods decrease gradually as $f_{b'}$ increases from 20 to 70. However, our method achieves the lowest failure rate compared to the RandCoEdge, DRLCoEdge, and SMCoEdge methods. Additionally, the failure rate of our method only exhibits slight increments from 0 to 0.05\% as $f_{b'}$ decreases from 70 to 20. Moreover, these results achieve a near performance compared with the Optimal method.

Besides, Fig. \ref{Fig6.sub.4} shows that as $\tau_{n}$ increases from 0.5 to 1.0, the average failure rates of all five methods decrease gradually. However, our method's failure rate is the lowest among RandCoEdge, DRLCoEdge, and SMCoEdge.

\subsubsection{Further Performance Evaluation on Our Method}\label{waiting-computing-transmission-delay-subsection}
Compared to our previous conference article \cite{xu2024enhancing}, this subsection evaluates the performance of our AMCoEdge method in terms of waiting delay, computing delay, and transmission delay across different task sizes, numbers of tasks, and ES computing capacities, as shown in Fig. \ref{Fig7}.

\begin{figure*}[!t]
\centering 
\subfigure[$d_{n}\sim\lbrack2, 5\rbrack$ \& $f_{b}\sim\lbrack10, 50\rbrack$]{
    \label{Fig7.sub.1}
    \includegraphics[width=1.7in]{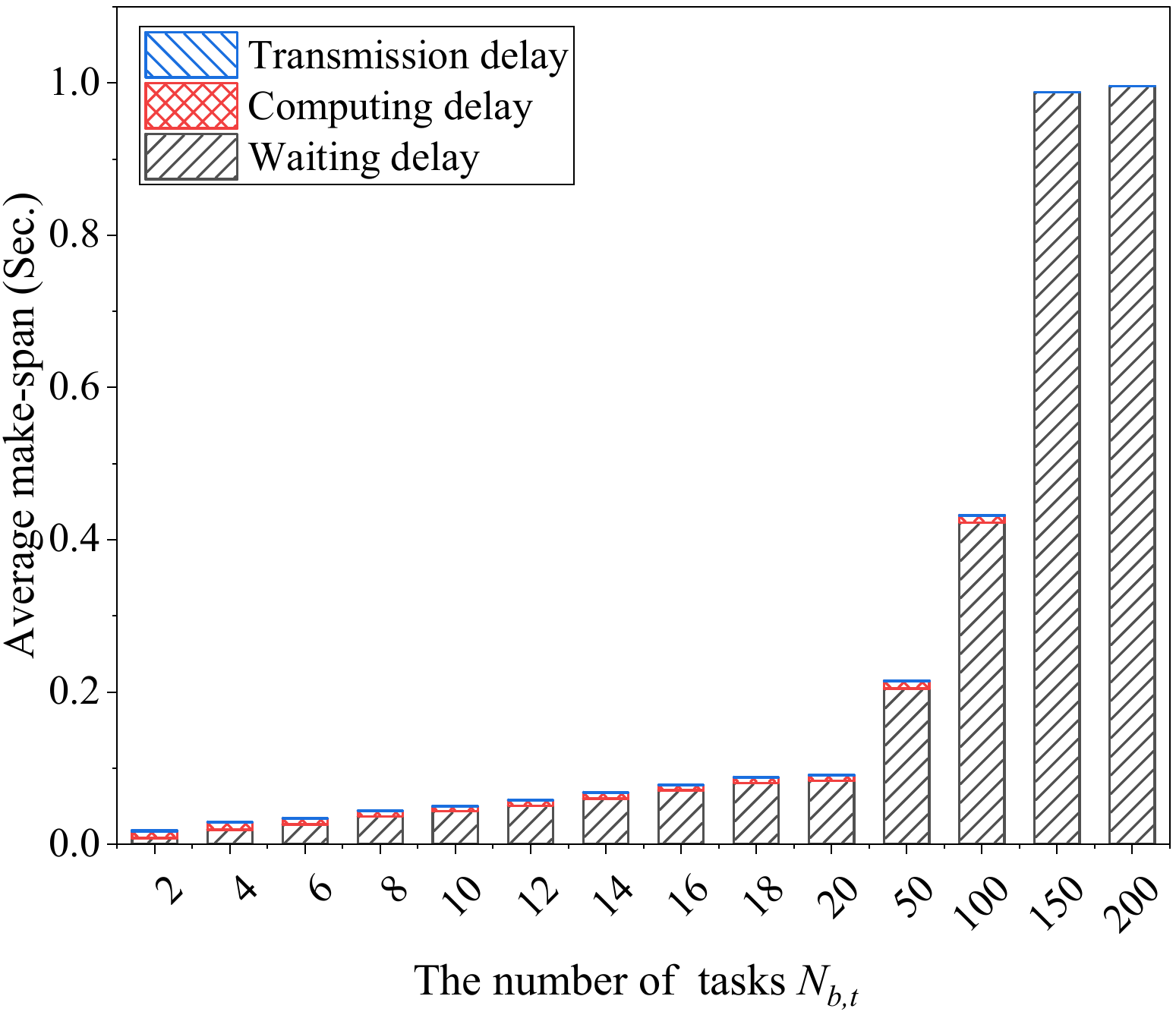}}
\subfigure[$d_{n}\sim\lbrack20, 50\rbrack$ \& $f_{b}\sim\lbrack10, 50\rbrack$]{
    \label{Fig7.sub.2}
    \includegraphics[width=1.7in]{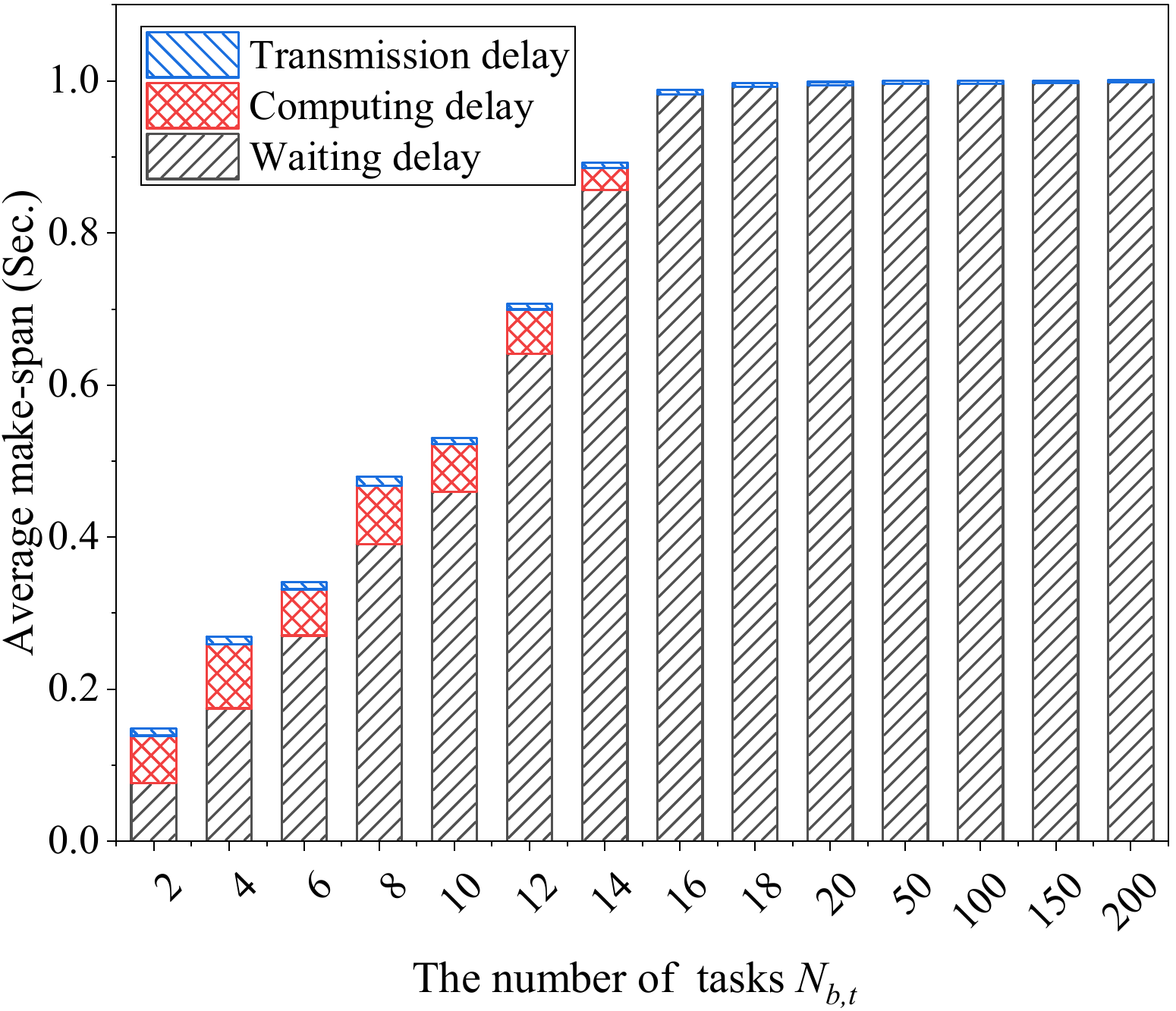}}
    \subfigure[$d_{n}\sim\lbrack2, 5\rbrack$ \& $N_{b,t}= 50$]{
    \label{Fig7.sub.3}
    \includegraphics[width=1.7in]{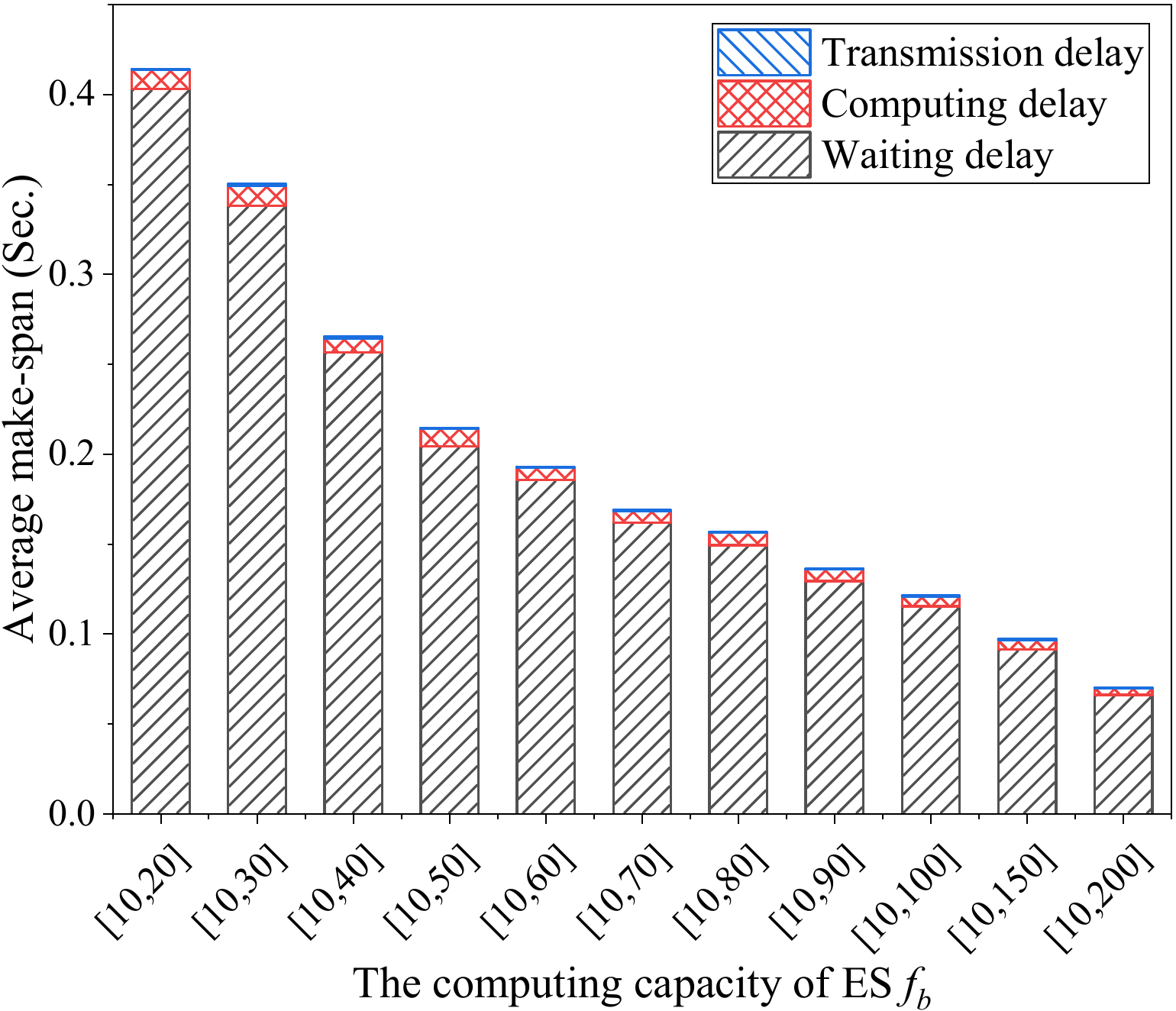}}
    \subfigure[$d_{n}\sim\lbrack20, 50\rbrack$ \& $N_{b,t} = 10$]{
    \label{Fig7.sub.4}
    \includegraphics[width=1.7in]{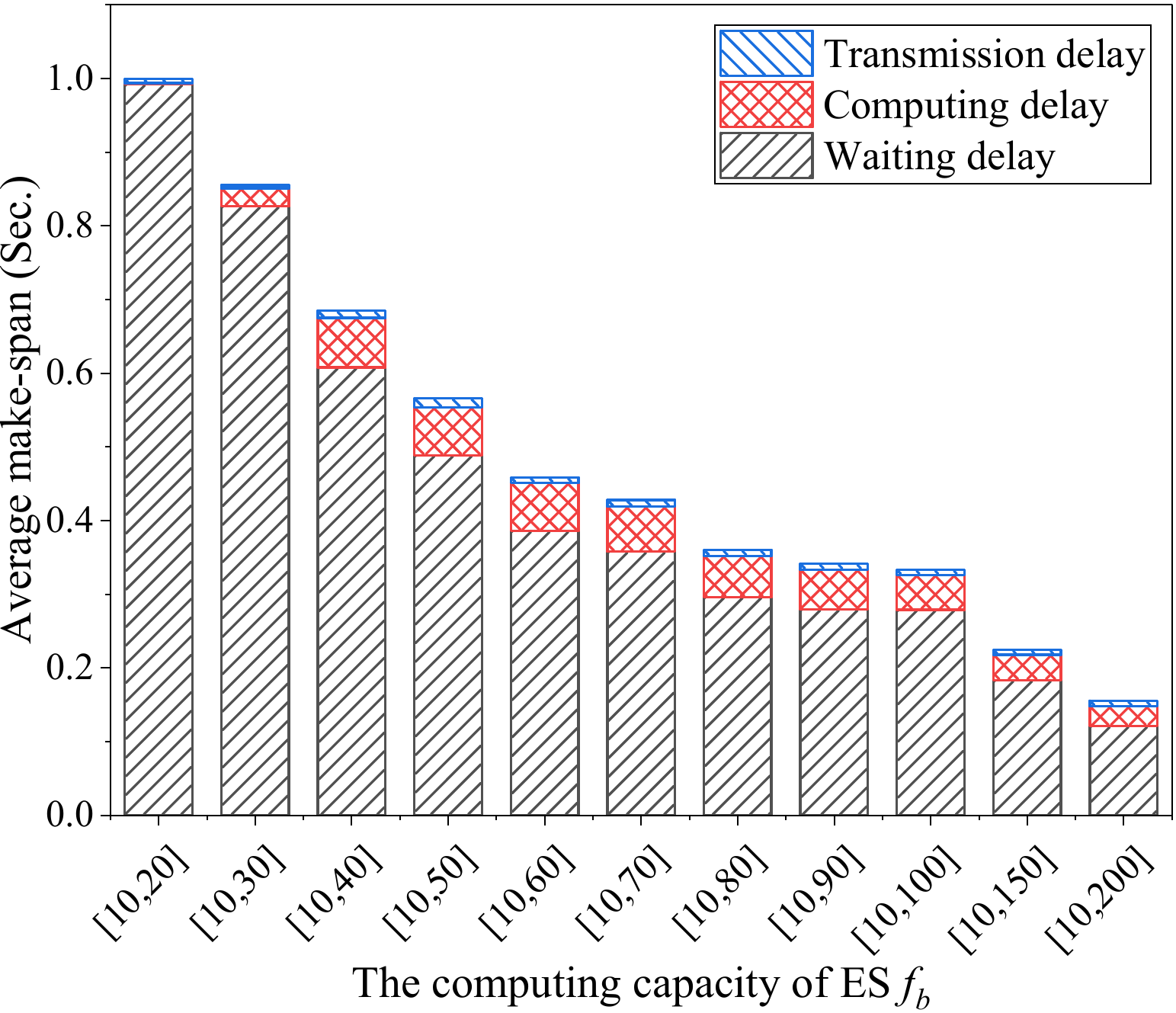}}
\caption{The performance of our method in terms of the waiting delay, computing delay, and transmission delay by varying different task sizes $d_{n}$, number of tasks $N_{b,t}$, and computing capacities of ESs $f_{b}$.}
\label{Fig7}
 % \vspace{-0.1cm}
\end{figure*}

Figs. \ref{Fig7.sub.1} and \ref{Fig7.sub.2} present that the average offloading make-span of the proposed method increases with the number of tasks $N_{b,t}$ ranging from 2 to 200, exhibiting a trend consistent with that observed in Fig. \ref{Fig5.sub.1}. Notably, the make-span stabilizes when $N_{b,t}$ becomes sufficiently large, as the service delay approaches the task offloading deadline $\tau_{n}$, beyond which tasks cannot be completed. Furthermore, when the task size $d_n$ increases from the range [2, 5] Mbits to [20, 50] Mbits, the number of tasks reaching the deadline decreases significantly from approximately 150 to 20. This reduction occurs because larger task sizes limit the number of tasks an ES can process effectively. In addition, as the number of tasks grows, the proportion of waiting delay in the total offloading makespan also increases, eventually dominating the overall task completion time. This behavior aligns with expectations, since a higher number of tasks leads to longer queuing times before processing.

From Figs. \ref{Fig7.sub.3} and \ref{Fig7.sub.4}, it is observed that the average make-span gradually decreases as the ES computing capacity $f_b$ increases from the range [10, 20] GHz to [10, 200] GHz, which aligns with the trends presented in Fig. \ref{Fig5.sub.3}. Moreover, under a smaller task size and a larger number of tasks, the proportion of waiting delay in the total make-span becomes more pronounced, indicating that system congestion due to task volume has a dominant effect under such conditions.

\begin{figure}[!t]
\centering 
\subfigure[Make-span comparisons]{
    \label{Fig8.sub.1}
    \includegraphics[width=0.48\linewidth]{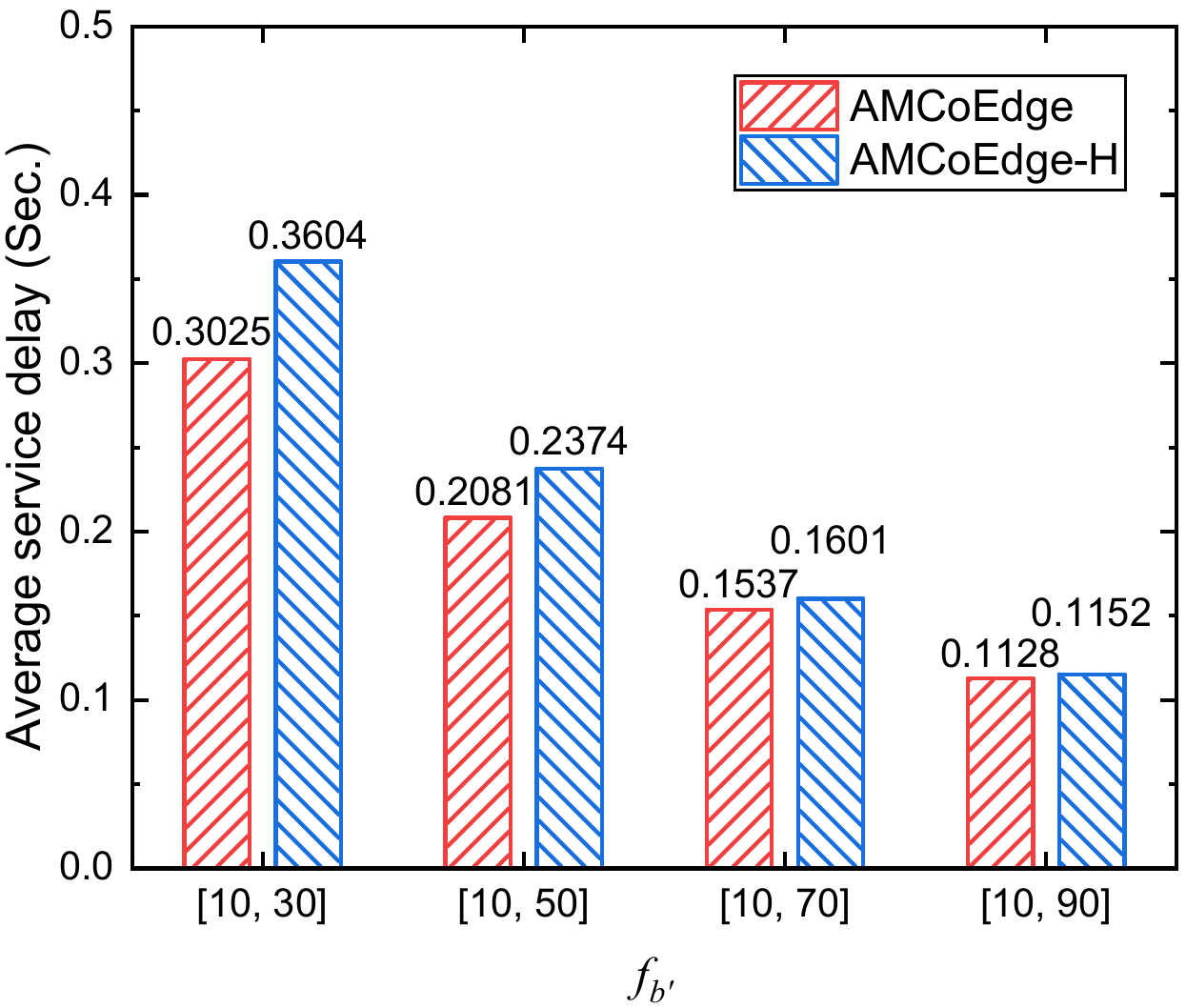}}
\subfigure[Failure rate comparisons]{
    \label{Fig8.sub.2}
    \includegraphics[width=0.48\linewidth]{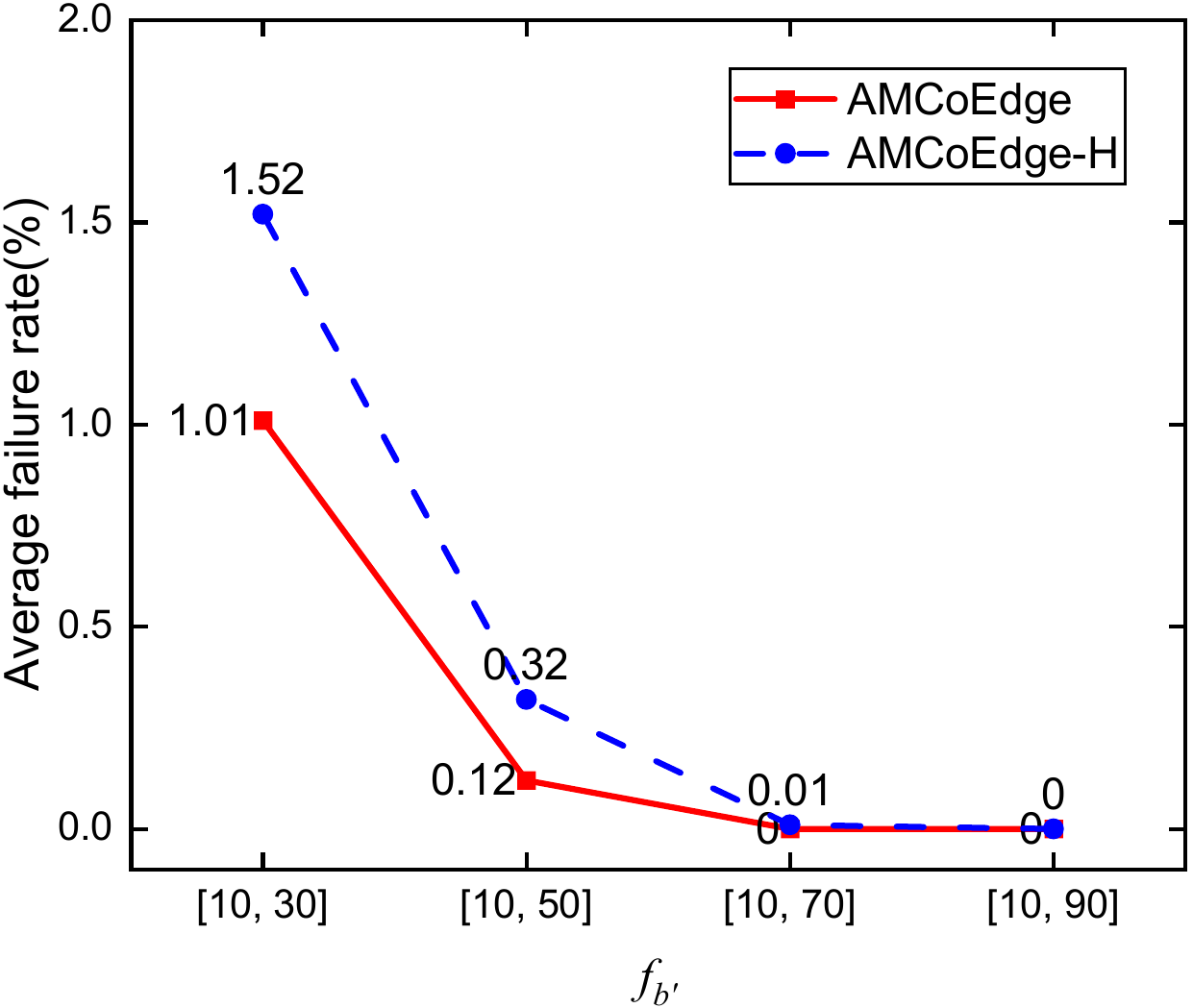}}
\caption{Offloading Make-span and failure rate comparisons of the proposed AMCoEdge and AMCoEdge-H methods by varying the ES's computing capacity $f_{b'}$.}
\label{Fig8}
\end{figure}

\subsubsection{The Cost-effectiveness of HECWA Algorithm}\label{cost-effectiveness subsection}
Compared with our previous conference article \cite{xu2024enhancing}, this section presents a performance comparison of the proposed CWA and HECWA algorithms to validate their effectiveness. The experimental results are shown in Fig. \ref{Fig8}, where the AMCoEdge-H method represents using the HECWA algorithm instead of the CWA algorithm in the Algorithm \ref{AMCoEdge-algorithm}.

From Fig. \ref{Fig8.sub.1}, as the computing capacity $f_{b}$ of ESs increases from $[10, 30]$ to $[10, 90]$, the average offloading make-span of the AMCoEdge-H method is always higher than that of the AMCoEdge method. However, using the HECWA algorithm, the AMCoEdge-H's average make-spans increase by only 0.0064 seconds and 0.0024 seconds, respectively, when the $f_{b}$ is $[10, 70]$ and $[10, 90]$, compared to using the CWA algorithm. Moreover, Fig. \ref{Fig8.sub.2} demonstrates that when the $f_{b}$ equals $[10, 70]$ and $[10, 90]$, the AMCoEdge and AMCoEdge-H methods have nearly equal failure rates. Importantly, by using the HECWA algorithm, the time complexity for the workload allocation is vastly reduced from $O(k^{3})$ to $O(k)$ (seen in \textbf{Theorems} \ref{CWA-complexity-Theorem} and \ref{HECWA-complexity-Theorem}). In addition, we also record the running time (except training time) of our Algorithm \ref{AMCoEdge-algorithm} for a task to \textbf{0.18} milliseconds and \textbf{0.17} milliseconds when the CWA and HECWA algorithms are used, respectively. That is to say, the overall efficiency of the Algorithm \ref{AMCoEdge-algorithm} is improved by \textbf{5.56\%} by using the HECWA algorithm. These results suggest that using the HECWA algorithm is cost-effective, especially for large ES computation capacities. This is because the HECWA algorithm neglects task processing time during workload allocation, resulting in a worse solution than the CWA algorithm. However, as ES capacities increase, waiting time decreases. Then, the HECWA achieves a better solution for the workload allocation stage, thereby reducing the task offloading make-span and failure rate.
 
In summary, the above results provide compelling evidence of our AMCoEdge method’s substantial superiority over state-of-the-art baselines in reducing task offloading, make-span, and failure rate.

\section{Prototype Implementation}\label{implementation}
This section implements a distributed edge prototype system and evaluates the practicality of our AMCoEdge method within it. 

\subsection{AIGC Model Deployment}
\subsubsection{System Implementation} As depicted in Fig. \ref{Fig9}, we develop a distributed edge computing prototype system with five Jetsons consisting of three AGX Orin devices and two Xavier NX devices. The AGX Orin device features an NVIDIA Ampere architecture GPU, 64 GB of memory, 64 Tensor Cores, and 2048 CUDA cores. The Xavier NX device features a NVIDIA Volta architecture GPU with 16 GB of memory, 48 Tensor Cores, and 384 CUDA cores. The five Jetsons serve as the five ESs in the system. Each ES deploys a scheduler integrated with our Algorithm \ref{AMCoEdge-algorithm}. These devices are interconnected via a wired Gigabit local-area network, enabling the system to process AIGC requests in parallel.

\subsubsection{AIGC Model Deployment.} For convenience, we use the SD model \cite{ho2020denoising} as a representative AIGC service. Specifically, the SD 1.5 model \cite{sd1_5} is deployed at each Jetson device, providing AIGC services for UDs.

\begin{figure}[!t]
    \centering
    \includegraphics[width=0.95\linewidth]{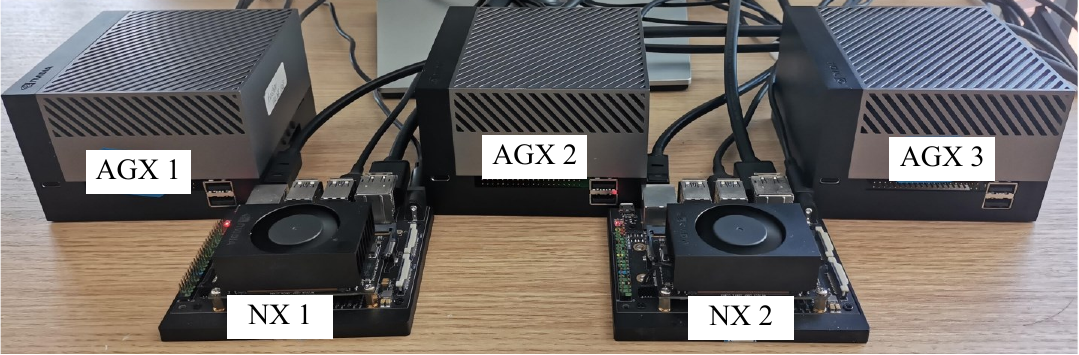}
    \caption{The developed prototype system with our AMCoEdge method.}
    \label{Fig9}
    % \vspace{-0.2cm}
\end{figure}

\subsection{Experimental Design}
To rapidly validate the practicality of our AMCoEdge method, we select the image-repairing task in AIGC services as a representative example. In particular, we use the publicly available \textbf{CelebA-HQ} dataset \cite{liu2015deep}. For each Jetson device $b$ at time slot $t$, $N_{b,t}$ images are randomly selected from the dataset, with a thin or a thick mask, for an AIGC task. Finally, the scheduler on each Jetson device employs our Algorithm \ref{AMCoEdge-algorithm} to offload the task to multiple Jetsons for collaborative processing, i.e., $\{N_{b,t} \cdot x_{b,n,t,b'}\}_{b'\in\mathcal{L}^{b,n}_{t}}$ images are allocated to the selected ESs $\mathcal{L}^{b,n}_{t}$ for parallel processing by the system. Here, each image-repairing process in the task is used as an independent subtask. Since the subtask (i.e., an image) cannot be partitioned, the fractions $\{N_{b,t} \cdot x_{b,n,t,b'}\}_{b'\in\mathcal{L}^{b,n}_{t}}$ of the image allocation must be an integer. We add an integer operation for the workload allocation procedure to satisfy this requirement. Specifically, if the allocation fraction $N_{b,t} \cdot x_{b,n,t,b'}$ equals max$\{N_{b,t} \cdot x_{b,n,t,b'}\}_{b'\in\mathcal{L}^{b,n}_{t}}$, the $N_{b,t} \cdot x_{b,n,t,b'}$ is set to $N_{b,t} - \sum_{v\in\mathcal{L}^{b,n}_{t}\setminus b'} round (N_{b,t} \cdot x_{b,n,t,v})$ and otherwise equals $round (N_{b,t} \cdot x_{b,n,t,b'})$. Here, $round(\cdot)$ is a function that rounds a decimal number to an integer. The variable $d_{n}$ represents the image length. We consider that the image size is $512\times512$ pixels. The deadline for each image repair is 60 seconds; i.e., the task with $N_{b,t}$ images has a deadline of $60 \cdot N_{b,t}$ seconds. We calculate the processing delay $T^{\text{proc}}_{b,n,t,b'}$ with the elapsed time from initiating task transmission to receiving the processed result.

\begin{figure}[!t]
\centering 
\subfigure[Make-span comparisons]{
    \label{Fig10.sub.1}
    \includegraphics[width=0.475\linewidth]{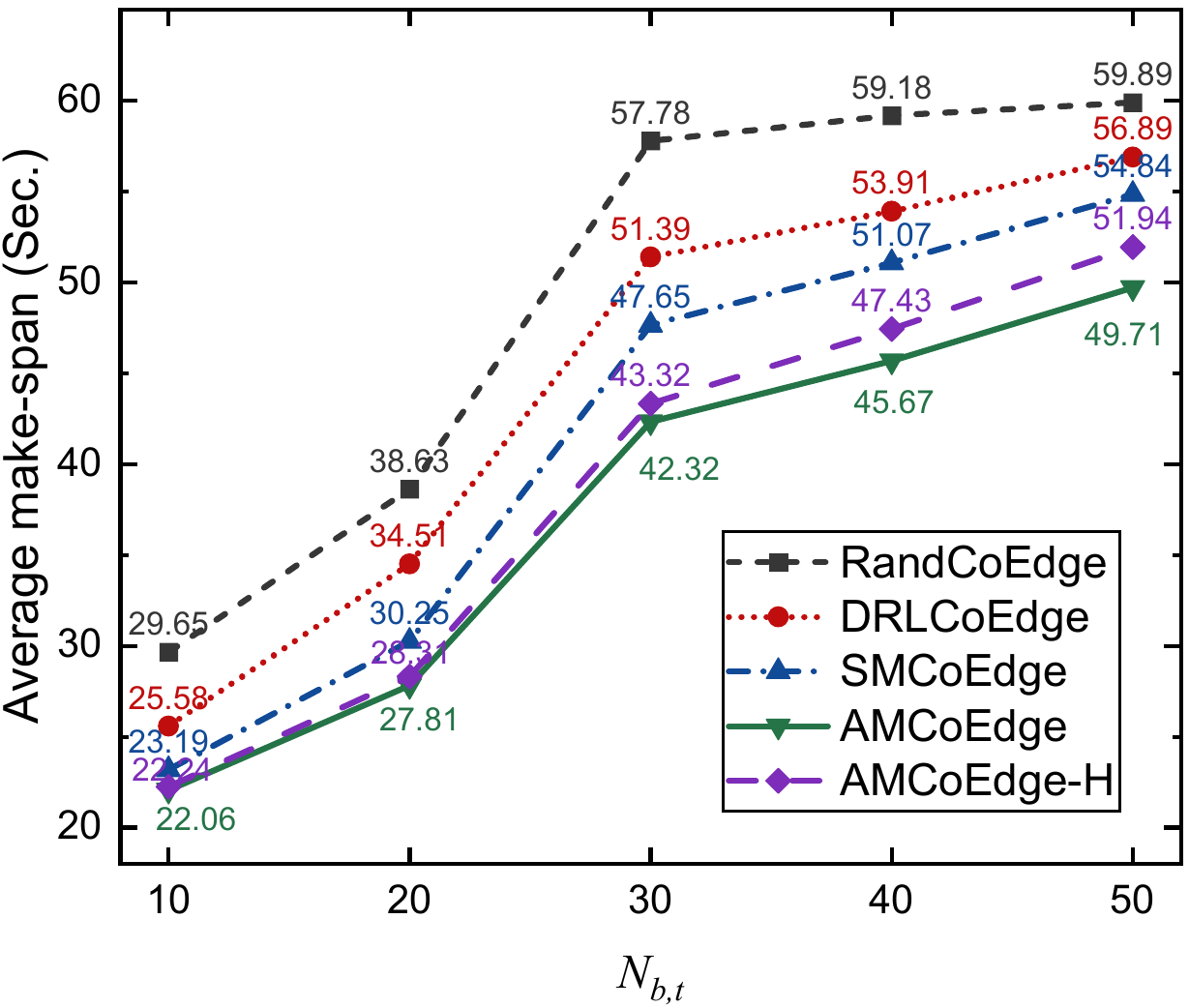}}
\subfigure[Failure rate comparisons]{
    \label{Fig10.sub.2}
    \includegraphics[width=0.48\linewidth]{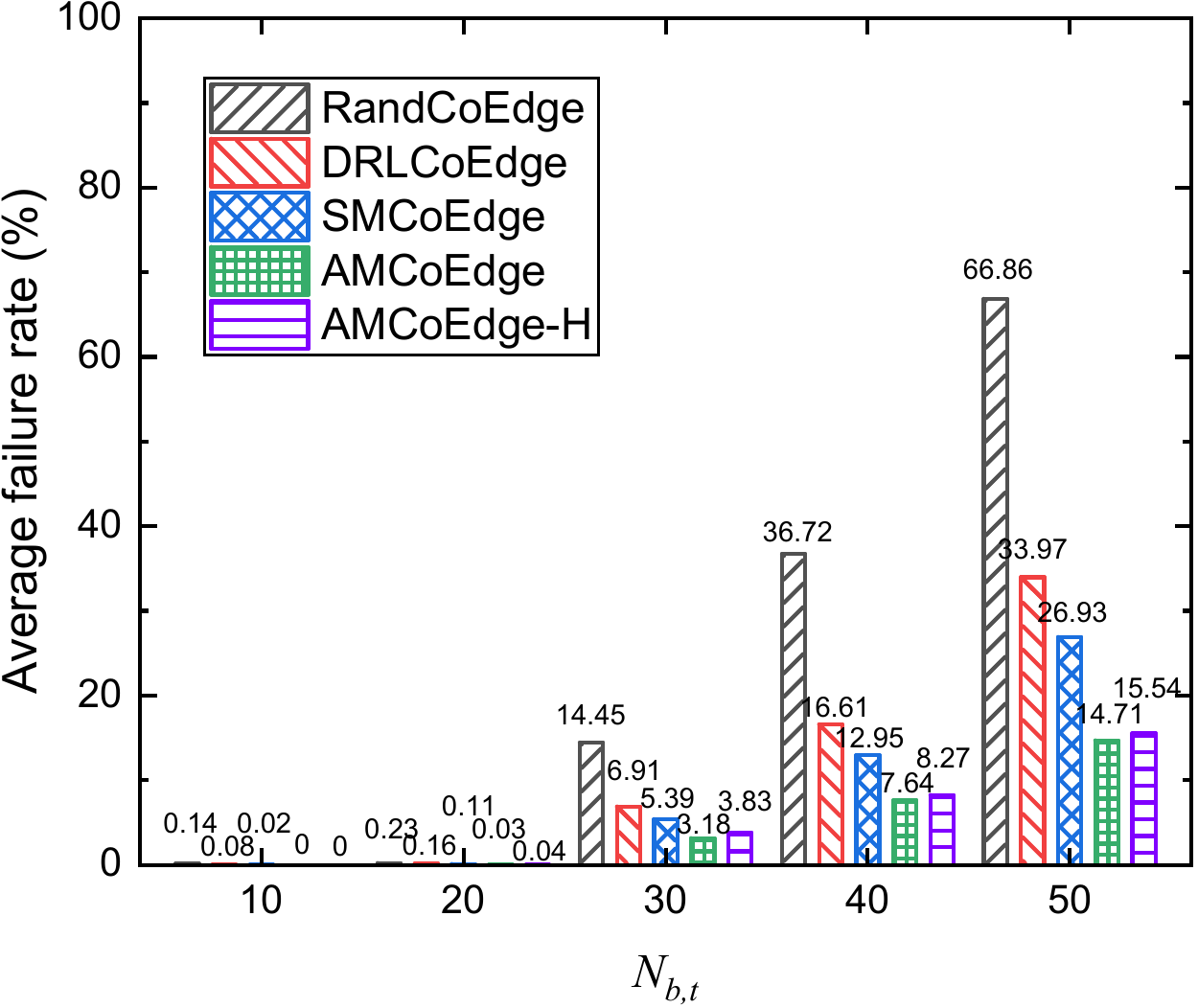}}
\caption{Offloading make-span and failure rate comparisons of our method and baselines by varying the number $N_{b,t}$ of tasks.}
\label{Fig10}
\end{figure}

\subsection{Test-bed Results}\label{test-bed-results}
We compare the average offloading make-spans and failure rates of our AMCoEdge and AMCoEdge-H methods with existing state-of-the-art methods across varying task counts, as detailed in Fig. \ref{Fig10}. Furthermore, partial results from the image-repairing tasks performed by our AMCoEdge method and prototype system are shown in Fig. \ref{Fig11}.

\subsubsection{Comparison Results}
As shown in Fig. \ref{Fig10}, all the methods' make-spans and failure rates increase with the increase of the number of tasks $N_{b,t}$. However, our AMCoEdge and AMCoEdge-H methods significantly outperform the other three methods in terms of average offloading make-span and failure rate across all values of $N_{b,t}$. Specifically, Fig. \ref{Fig10.sub.1} shows that when the $N_{b,t}$ increases from 10 to 50, the average make-spans of our AMCoEdge and AMCoEdge-H methods are increased from 22.06 to 49.71 seconds and 22.24 to 51.94 seconds, respectively. In contrast, the average make-spans of RandCoEdge, DRLCoEdge, and SMCoEdge methods are increased from 29.65 to 59.89 seconds, 25.58 to 56.89 seconds, and 23.19 to 54.84 seconds, respectively, which are higher than that of our AMCoEdge method by the average values \textbf{31.98\%}, \textbf{19.79\%}, and \textbf{9.23\%}, respectively. 

On the other hand, Fig. \ref{Fig10.sub.2} shows that the average failure rates for all methods also increase with the number of tasks. The reason is that as the number of input tasks increases, more tasks will be stored in the queue waiting for processing. Thus, the computation time for each task increases. Additionally, because the computing resources of all ESs are limited, more tasks will be dropped when their completion times exceed the deadline, resulting in more task failures. Besides, Fig. \ref{Fig10} presents that the average make-span and failure rate of the AMCoEdge method always approximate those of the AMCoEdge-H method. This result again validates the cost-effectiveness of the proposed HECWA algorithm.

\subsubsection{Partial repairing Results in the Prototype System}
We also examine the image-repairing quality of our AMCoEdge method in the prototype system. The selective results are depicted in Fig. \ref{Fig11}. We can see that these results meet the requirements of the image-repairing task well, validating the practical applicability of our method in a real system. However, this paper mainly focuses on the improvements in AIGC service efficiency rather than AIGC service quality.

\begin{figure}[!t]
    \centering
    \includegraphics[width=\linewidth]{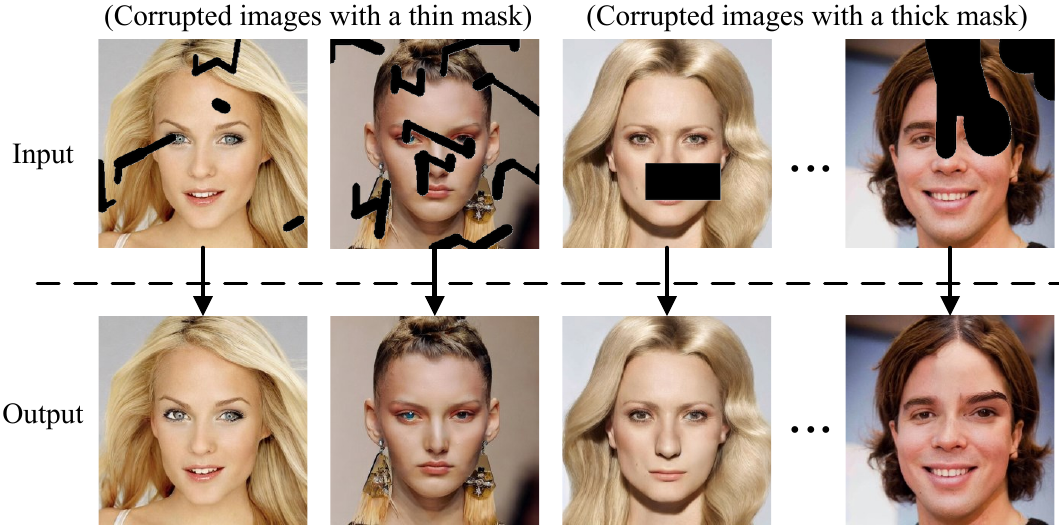}
    \caption{The partial results of image repairing by our method and system.}
    \label{Fig11}
\end{figure}

\section{Discussion}\label{discussion}
See Appendix B \cite{xu2024enhancing, fan2024collaborative, zhou2022two, xu2024dynamic, huang2025drmq, chen2024task}. 

\section{Conclusion}\label{conclusion}
In this paper, we propose integrating collaborative MEC technology to reduce processing delays for AIGC services and present an innovative AMCoEdge method. The AMCoEdge performs adaptive multi-ES selection and online workload allocation to accelerate AIGC service in heterogeneous edge environments. The AMCoEdge method is implemented using a two-stage algorithm that enhances the efficiency of solving our problem using the proposed ADQN model and the CWA method. Meanwhile, we provide theoretical performance analysis to establish that our algorithm achieves a near-optimal solution to our problem with approximate linear-time complexity. Extensive simulation results conclusively demonstrate the superiority of our method over state-of-the-art approaches. Moreover, our method achieves a minimum 11.04\% reduction in task offloading make-span and a minimum 44.86\% decrease in task offloading failure rate. Furthermore, compared to our previous conference article \cite{xu2024enhancing}, we design an HECWA algorithm to quickly achieve the optimal solution for the workload allocation stage and validate HECWA's cost-effectiveness. We also develop a prototype system to implement and evaluate our AMCoEdge and AMCoEdge-H methods. Test-bed results reveal that the AMCoEdge achieves $9.23\% - 31.98\%$ improvements in task offloading make-span over existing representative methods. 

\bibliographystyle{IEEEtran}  % bibliography style
\bibliography{references} % bibliography references

@misc{huggingface,
    howpublished = {Hugging Face Platform: \url{https://huggingface.co/spaces}}    
}

@misc{sd1_5,
    howpublished = {Stable Diffusion 1.5 Model: \url{https://github.com/bubbliiiing/stable-diffusion}}    
}

@misc{midjourney,
    howpublished = {Midjourney Platform: \url{https://www.midjourney.com/home}}    
}

@article{cao2023comprehensive,
author={Cao, Y and Li, S and Liu, Y and Yan, Z and Dai, Y and Yu, P S and Sun, L},
title={A comprehensive survey of ai-generated content (aigc): A history of generative ai from gan to chatgpt},
journal={arXiv preprint arXiv:2303.04226},
pages={1--44},
year={2023},
publisher={IEEE}
}

@article{chen2024task,
  title={Task Offloading and Resource Pricing Based on Game Theory in UAV-Assisted Edge Computing},
  author={Chen, Zhuoyue and Yang, Yaozong and Xu, Jiajie and Chen, Ying and Huang, Jiwei},
  journal={IEEE Transactions on Services Computing},
  year={2024},
  publisher={IEEE}
}

@article{chen2021collaborative,
  title={Collaborative service placement for edge computing in dense small cell networks},
  author={Chen, Lixing and Shen, Cong and Zhou, Pan and Xu, Jie},
  journal={IEEE Transactions on Mobile Computing},
  volume={20},
  number={2},
  pages={377--390},
  year={2021},
  publisher={IEEE}
}

@article{chu2023online,
  title={Online optimal service selection, resource allocation and task offloading for multi-access edge computing: A utility-based approach},
  author={Chu, Weibo and Yu, Peijie and Yu, Zhiwen and Lui, John CS and Lin, Yi},
  journal={IEEE Transactions on Mobile Computing},
  volume={22},
  number={7},
  pages={4150--4167},
  year={2023},
  publisher={IEEE}
}

@article{du2024exploring,
  title={Exploring collaborative distributed diffusion-based AI-generated content (AIGC) in wireless networks},
  author={Du, Hongyang and Zhang, Ruichen and Niyato, Dusit and Kang, Jiawen and Xiong, Zehui and Kim, Dong In and Shen, Xuemin Sherman and Poor, H Vincent},
  journal={IEEE Network},
  year={2024},
  volume={38},
  number={3},
  pages={178-186},
  publisher={IEEE}
}

@article{du2023ai,
  title={AI-generated incentive mechanism and full-duplex semantic communications for information sharing},
  author={Du, Hongyang and Wang, Jiacheng and Niyato, Dusit and Kang, Jiawen and Xiong, Zehui and Kim, Dong In},
  journal={IEEE Journal on Selected Areas in Communications},
  year={2023},
  publisher={IEEE}
}

@article{du2023aigenerated,
  author={Du, Hongyang and Wang, Jiacheng and Niyato, Dusit and Kang, Jiawen and Xiong, Zehui and Kim, Dong In},
  journal={IEEE Journal on Selected Areas in Communications}, 
  title={AI-Generated Incentive Mechanism and Full-Duplex Semantic Communications for Information Sharing}, 
  year={2023},
  volume={41},
  number={9},
  pages={2981-2997}
}

@article{fan2024collaborative,
  title={Collaborative service placement, task scheduling, and resource allocation for task offloading with edge-cloud cooperation},
  author={Fan, Wenhao and Zhao, Liang and Liu, Xun and Su, Yi and Li, Shenmeng and Wu, Fan and Liu, Yuan'an},
  journal={IEEE Transactions on Mobile Computing},
  volume={23},
  number={1},
  pages={238-256},
  year={2024},
  publisher={IEEE}
}

@article{farhadi2021service,
  title={Service placement and request scheduling for data-intensive applications in edge clouds},
  author={Farhadi, Vajiheh and Mehmeti, Fidan and He, Ting and La Porta, Thomas F and Khamfroush, Hana and Wang, Shiqiang and Chan, Kevin S and Poularakis, Konstantinos},
  journal={IEEE/ACM Transactions on Networking},
  volume={29},
  number={2},
  pages={779--792},
  year={2021},
  publisher={IEEE}
}

@article{gao2023task,
  author={Gao, Mingjin and Shen, Rujing and Shi, Long and Qi, Wen and Li, Jun and Li, Yonghui},
  journal={IEEE Transactions on Mobile Computing}, 
  title={Task Partitioning and Offloading in DNN-Task Enabled Mobile Edge Computing Networks}, 
  year={2023},
  volume={22},
  number={4},
  pages={2435-2445}
}

@inproceedings{gu2021layer,
  title={Layer aware microservice placement and request scheduling at the edge},
  author={Gu, Lin and Zeng, Deze and Hu, Jie and Li, Bo and Jin, Hai},
  booktitle={IEEE INFOCOM 2021-IEEE Conference on Computer Communications},
  pages={1--9},
  year={2021},
  organization={IEEE}
}

@article{he2023priority,
  title={Priority-based offloading optimization in cloud-edge collaborative computing},
  author={He, Zhenli and Xu, Yanan and Zhao, Mingxiong and Zhou, Wei and Li, Keqin},
  journal={IEEE Transactions on Services Computing},
  volume={16},
  number={06},
  pages={3906--3919},
  year={2023},
  publisher={IEEE Computer Society}
}

@inproceedings{han2021tailored,
  title={Tailored learning-based scheduling for kubernetes-oriented edge-cloud system},
  author={Han, Yiwen and Shen, Shihao and Wang, Xiaofei and Wang, Shiqiang and Leung, Victor CM},
  booktitle={IEEE INFOCOM 2021-IEEE Conference on Computer Communications},
  pages={1--10},
  year={2021},
  organization={IEEE}
}

@inproceedings{he2016deep,
  title={Deep residual learning for image recognition},
  author={He, Kaiming and Zhang, Xiangyu and Ren, Shaoqing and Sun, Jian},
  booktitle={Proceedings of the IEEE conference on computer vision and pattern recognition},
  pages={770--778},
  year={2016}
}

@article{huang2022throughput,
  title={Throughput guarantees for multi-cell wireless powered communication networks with non-orthogonal multiple access},
  author={Huang, Liang and Nan, Runkai and Chi, Kaikai and Hua, Qiaozhi and Yu, Keping and Kumar, Neeraj and Guizani, Mohsen},
  journal={IEEE Transactions on Vehicular Technology},
  volume={71},
  number={11},
  pages={12104--12116},
  year={2022},
  publisher={IEEE}
}

@article{ho2020denoising,
  title={Denoising diffusion probabilistic models},
  author={Ho, Jonathan and Jain, Ajay and Abbeel, Pieter},
  journal={Advances in Neural Information Processing Systems},
  volume={33},
  pages={6840--6851},
  year={2020}
}

@article{huang2025drmq,
  title={DRMQ: Dynamic Resource Management for Enhanced QoS in Collaborative Edge-Edge Industrial Environments},
  author={Huang, Fengyi and Wang, Wenhua and Liu, Qin and Fan, Wentao and Guo, Jianxiong and Jia, Weijia and Cao, Jiannong and Wang, Tian},
  journal={IEEE Transactions on Services Computing},
  year={2025},
  publisher={IEEE}
}

@inproceedings{huang2022towards,
  author={Huang, Lei and Liang, Zhiying and Sreekumar, Nikhil and Kaushik, Sumanth and Chandra, Abhishek and Weissman, Jon},
  booktitle={2022 IEEE 42nd International Conference on Distributed Computing Systems (ICDCS)}, 
  title={Towards Elasticity in Heterogeneous Edge-dense Environments}, 
  year={2022},
  volume={},
  number={},
  pages={403-413},
  doi={10.1109/ICDCS54860.2022.00046}}

@article{li2020deep,
  title={Deep reinforcement learning for collaborative edge computing in vehicular networks},
  author={Li, Mushu and Gao, Jie and Zhao, Lian and Shen, Xuemin},
  journal={IEEE Transactions on Cognitive Communications and Networking},
  volume={6},
  number={4},
  pages={1122--1135},
  year={2020},
  publisher={IEEE}
}

@article{lin2024blockchain,
  author={Lin, Yijing and Gao, Zhipeng and Du, Hongyang and Niyato, Dusit and Kang, Jiawen and Xiong, Zehui and Zheng, Zibin},
  journal={IEEE Transactions on Services Computing}, 
  title={Blockchain-Based Efficient and Trustworthy AIGC Services in Metaverse}, 
  year={2024},
  volume={17},
  number={5},
  pages={2067-2079}
}

@article{liang2025collaborative,
  title={Collaborative Edge Server Placement for Maximizing QoS with Distributed Data Cleaning},
  author={Liang, Yuzhu and Yin, Mujun and Wang, Wenhua and Liu, Qin and Wang, Liang and Zheng, Xi and Wang, Tian},
  journal={IEEE Transactions on Services Computing},
  year={2025},
  publisher={IEEE}
}

@inproceedings{liu2015deep,
  author={Liu, Ziwei and Luo, Ping and Wang, Xiaogang and Tang, Xiaoou},
  booktitle={2015 IEEE International Conference on Computer Vision (ICCV)}, 
  title={Deep Learning Face Attributes in the Wild}, 
  year={2015},
  pages={3730-3738},
  organization={IEEE}
}

@article{lin2024game,
  title={A Game-based Computation Offloading with Imperfect Information in Multi-Edge Environments},
  author={Lin, Bing and Weng, Jie and Chen, Xing and Ma, Yun and Hsu, Ching-Hsien},
  journal={IEEE Transactions on Services Computing},
  year={2024},
  publisher={IEEE}
}

@inproceedings{ma2020cooperative,
  title={Cooperative service caching and workload scheduling in mobile edge computing},
  author={Ma, Xiao and Zhou, Ao and Zhang, Shan and Wang, Shangguang},
  booktitle={IEEE INFOCOM 2020-IEEE Conference on Computer Communications},
  pages={2076--2085},
  year={2020},
  organization={IEEE}
}

@article{maaz2023video,
  title={Video-chatgpt: Towards detailed video understanding via large vision and language models},
  author={Maaz, Muhammad and Rasheed, Hanoona and Khan, Salman and Khan, Fahad Shahbaz},
  journal={arXiv preprint arXiv:2306.05424},
  year={2023}
}

@article{mnih2015human,
  title={Human-level control through deep reinforcement learning},
  author={Mnih, Volodymyr and Kavukcuoglu, Koray and Silver, David and Rusu, Andrei A and Veness, Joel and Bellemare, Marc G and Graves, Alex and Riedmiller, Martin and Fidjeland, Andreas K and Ostrovski, Georg and others},
  journal={nature},
  volume={518},
  number={7540},
  pages={529--533},
  year={2015},
  publisher={Nature Publishing Group}
}

@article{sthapit2018computational,
  title={Computational load balancing on the edge in absence of cloud and fog},
  author={Sthapit, Saurav and Thompson, John and Robertson, Neil M and Hopgood, James R},
  journal={IEEE Transactions on Mobile Computing},
  volume={18},
  number={7},
  pages={1499--1512},
  year={2018},
  publisher={IEEE}
}

@article{sun2024edgebrain,
  author={Sun, Hui and Wang, Zhiyong and Yu, Ying and Sha, Kewei and Wu, Yalong},
  journal={IEEE Transactions on Services Computing}, 
  title={EdgeBrain: A Game Based Collaborative Computational Task Offloading Framework for Edge Video Analytics}, 
  year={2024},
  volume={17},
  number={5},
  pages={2287-2303}
}

@article{tang2022deep,
  title={Deep reinforcement learning for task offloading in mobile edge computing systems},
  author={Tang, Ming and Wong, Vincent WS},
  journal={IEEE Transactions on Mobile Computing},
  volume={21},
  number={6},
  pages={1985--1997},
  year={2022},
  publisher={IEEE}
}

@article{van2023chatgpt,
author={Van Dis, E. A. and Bollen, J. and Zuidema, W. and Van Rooij, R. and C. L. Bockting},
title={ChatGPT: five priorities for research},
journal={Nature},
volume={614},
number={7947},
pages={224--226},
year={2023}
}

@article{wang2023edge,
  title={Edge Computing and Sensor-Cloud: Overview, Solutions, and Directions},
  author={Wang, Tian and Liang, Yuzhu and Shen, Xuewei and Zheng, Xi and Mahmood, Adnan and Sheng, Quan Z},
  journal={ACM Computing Surveys},
  volume={55},
  number={13},
  pages={1--37},
  year={2023},
  publisher={ACM New York, NY}
}

@article{wang2023next,
  author={Wang, Shangshang and Shao, Ziyu and Lui, John C.S.},
  journal={IEEE Transactions on Mobile Computing}, 
  title={Next-Word Prediction: A Perspective of Energy-Aware Distributed Inference}, 
  year={2023},
  volume={},
  number={},
  pages={1-14},
  doi={10.1109/TMC.2023.3310536}
}

@article{xu2024dynamic,
  title={Dynamic Parallel Multi-Server Selection and Allocation in Collaborative Edge Computing},
  author={Xu, Changfu and Guo, Jianxiong and Li, Yupeng and Zou, Haodong and Jia, Weijia and Wang, Tian},
  journal={IEEE Transactions on Mobile Computing},
  year={2024},
  volume={23},
  number={11},
  pages={10523-10537},
  publisher={IEEE}
}

@inproceedings{xu2023smcoedge,
  title={SMCoEdge: Simultaneous multi-server offloading for collaborative mobile edge computing},
  author={Xu, Changfu and Li, Yupeng and Chu, Xiaowen and Zou, Haodong and Jia, Weijia and Wang, Tian},
  booktitle={The 23rd International Conference on Algorithms and Architectures for Parallel Processing (ICA3PP)},
  pages={1--18},
  year={2023},
  organization={Springer}
}

@inproceedings{xu2024enhancing, 
title={Enhancing AI-Generated Content Efficiency through Adaptive Multi-Edge Collaboration}, author={Xu, Changfu and Guo, Jianxiong and Zeng, Jiandian and Meng, Shengguang and Chu, Xiaowen and Cao, Jiannong and Wang, Tian}, 
booktitle={2024 IEEE 44th International Conference on Distributed Computing Systems (ICDCS)}, 
pages={960-970}, 
year={2024}, 
publisher={IEEE}
}

@ARTICLE{xu2023sparks,
  author={Xu, Minrui and Niyato, Dusit and Zhang, Hongliang and Kang, Jiawen and Xiong, Zehui and Mao, Shiwen and Han, Zhu},
  journal={IEEE Vehicular Technology Magazine}, 
  title={Sparks of Generative Pretrained Transformers in Edge Intelligence for the Metaverse: Caching and Inference for Mobile Artificial Intelligence-Generated Content Services}, 
  year={2023},
  volume={18},
  number={4},
  pages={35-44},
  organization={IEEE}
}

@article{xu2024unleashing,
  title={Unleashing the Power of Edge-Cloud Generative AI in Mobile Networks: A Survey of AIGC Services}, 
  author={Xu, Minrui and Du, Hongyang and Niyato, Dusit and Kang, Jiawen and Xiong, Zehui and Mao, Shiwen and Han, Zhu and Jamalipour, Abbas and Kim, Dong In and Shen, Xuemin and Leung, Victor C. M. and Poor, H. Vincent},
  journal={IEEE Communications Surveys \& Tutorials}, 
  year={2024},
  volume={26},
  number={2},
  pages={1127-1170},
}

@inproceedings{ye2024galaxy,
  title={Galaxy: A Resource-Efficient Collaborative Edge AI System for In-situ Transformer Inference},
  author={Ye, Shengyuan and Du, Jiangsu and Zeng, Liekang and Ou, Wenzhong and Chu, Xiaowen and Lu, Yutong and Chen, Xu},
  booktitle={IEEE INFOCOM 2024-IEEE Conference on Computer Communications},
  pages={1--10},
  year={2024},
  organization={IEEE}
}

@article{Yu2021when,
  title={When deep reinforcement learning meets federated learning: intelligent multitimescale resource management for multiaccess edge computing in 5G ultradense network},
  author={Yu, Shuai and Chen, Xu and Zhou, Zhi and Gong, Xiaowen and Wu, Di}, 
  journal={IEEE Internet of Things Journal}, 
  volume={8},
  number={4},
  pages={2238-2251},
  year={2021}, 
  publisher={IEEE}
}

@article{zeng2021coedge,
  title={CoEdge: Cooperative DNN Inference With Adaptive Workload Partitioning Over Heterogeneous Edge Devices},
  author={Zeng, Liekang and Chen, Xu and Zhou, Zhi and Yang, Lei and Zhang, Junshan},
  journal={IEEE/ACM Transactions on Networking},
  volume={29},
  number={2},
  pages={595-608},
  year={2021},
  publisher={IEEE}
  }

@ARTICLE{zhang2023hybrid,
  author={Zhang, Jiangjiang and Gong, Bei and Waqas, Muhammad and Tu, Shanshan and Han, Zhu},
  journal={IEEE Transactions on Services Computing}, 
  title={A Hybrid Many-Objective Optimization Algorithm for Task Offloading and Resource Allocation in Multi-Server Mobile Edge Computing Networks}, 
  year={2023},
  volume={16},
  number={5},
  pages={3101-3114}
}

@inproceedings{zhou2022two,
  title={Two Time-Scale Joint Service Caching and Task Offloading for UAV-assisted Mobile Edge Computing},
  author={Zhou, Ruiting and Wu, Xiaoyi and Tan, Haisheng and Zhang, Renli},
  booktitle={IEEE INFOCOM 2022-IEEE Conference on Computer Communications},
  pages={1189--1198},
  year={2022},
  organization={IEEE}
}

@ARTICLE{zou2025logical,
  author={Zou, Haodong and Guo, Jianxiong and Li, Yupeng and Fan, Wentao and Su, Weifeng and Xu, Changfu and Liang, Yuzhu and Wang, Tian and Cao, Jiannong},
  journal={IEEE Transactions on Mobile Computing}, 
  title={Logical Correction Enabled Collaborative Person Detection Inference in Edge Networks}, 
  year={2025, doi: 10.1109/TMC.2025.3626734},
  volume={},
  number={},
  pages={1-15},
  doi={10.1109/TMC.2025.3626734}
}

\begin{IEEEbiography}
[{\includegraphics[width=1in,height=1.25in,clip,keepaspectratio]{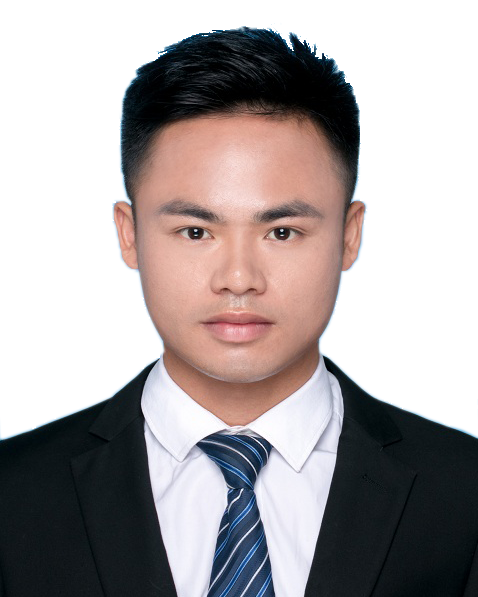}}]
{Changfu Xu}
(Member, IEEE) received the B.S. degree in Communication Engineering and the M.S. degree in Software Engineering from Jiangxi University of Finance and Economics in 2015 and 2018, respectively. He received his Ph.D. degree in Computer Science and Technology from Hong Kong Baptist University, Hong Kong, China, in 2025. Currently, he is an Assistant Professor with the School of Software and IoT Engineering, Jiangxi University of Finance and Economics, Nanchang, China. His main research interests include edge computing and AIGC. He has published more than 10 papers.
\end{IEEEbiography}

\begin{IEEEbiography}
[{\includegraphics[width=1in,height=1.25in,clip,keepaspectratio]{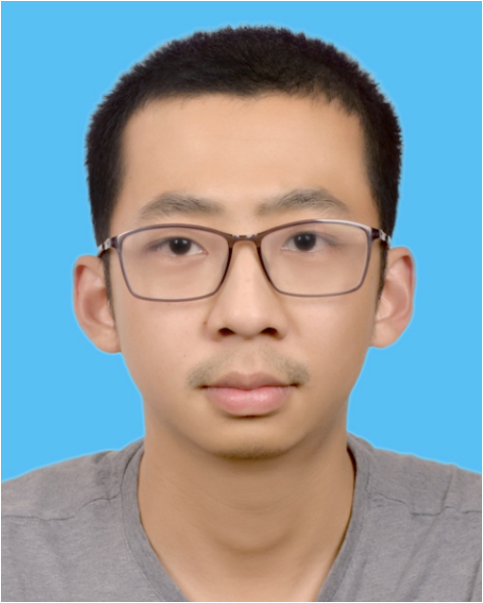}}]
{Jianxiong Guo}
(Member, IEEE) received his Ph.D. degree from the Department of Computer Science, The University of Texas at Dallas, in 2021, and his B.E. degree from the School of Chemistry and Chemical Engineering, South China University of Technology, in 2015. He is currently an associate professor at the Institute of AI and Future Networks, Beijing Normal University, China, and also at the Guangdong Key Lab of AI and Multi-Modal Data Processing, Beijing Normal-Hong Kong Baptist University, China. He is a member of IEEE/ACM/CCF. He has published more than 120 peer-reviewed papers and has served as a reviewer for numerous international journals/conferences. His research interests include edge intelligence, IoT, distributed machine learning, and online/combinatorial optimization.
\end{IEEEbiography}

\begin{IEEEbiography}
[{\includegraphics[width=1in,height=1.25in,clip,keepaspectratio]{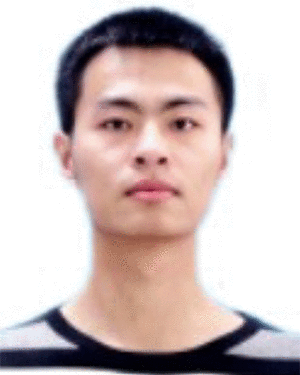}}]
{Jiandian Zeng}
(Member, IEEE) received the B.Sc. Degree in computer science and technology from Jianghan University, Wuhan, China, in 2014, and the M.Sc degree in computer technology from Huaqiao University, Xiamen, China, in 2018. He received the Ph.D. degree in Computer and Information Science from University of Macau, Macau, China, in 2023. Currently, he is a Lecturer with school of Computer Science and Technology, Beijing Normal University, Zhuhai, China. His research interests include natural language processing and mobile computing.
\end{IEEEbiography}

\begin{IEEEbiography}
[{\includegraphics[width=1in,height=1.25in,clip,keepaspectratio]{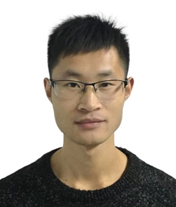}}]
{Houming Qiu} received the Ph.D. degree from the College of Computer Science and Technology, Nanjing University of Aeronautics and Astronautics, China. Currently, he is a Lecturer with the School of Software and IoT Engineering, Jiangxi University of Finance and Economics, Nanchang, China. His research interests include distributed computing, coding theory, and secure/private computing.
\end{IEEEbiography}

\begin{IEEEbiography}
[{\includegraphics[width=1in,height=1.25in,clip,keepaspectratio]{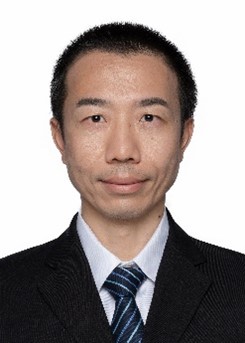}}]
{Tian Wang}
(Senior Member, IEEE) received his BSc and MSc degrees in Computer Science from Central South University in 2004 and 2007, respectively. He received his Ph.D. degree in Computer Science from the City University of Hong Kong in 2011. Currently, he is a professor at the Institute of Artificial Intelligence and Future Networks, Beijing Normal University. His research interests include Internet of Things, edge computing, and mobile computing. He has 30 patents and has published more than 200 papers in high-level journals and conferences. He was a co-recipient of the Best Paper Runner-up Award of IEEE/ACM IWQoS 2024. He has more than 17000 citations, according to Google Scholar. His H-index is 71.
\end{IEEEbiography}

\begin{IEEEbiography}
[{\includegraphics[width=1in,height=1.25in,clip,keepaspectratio]{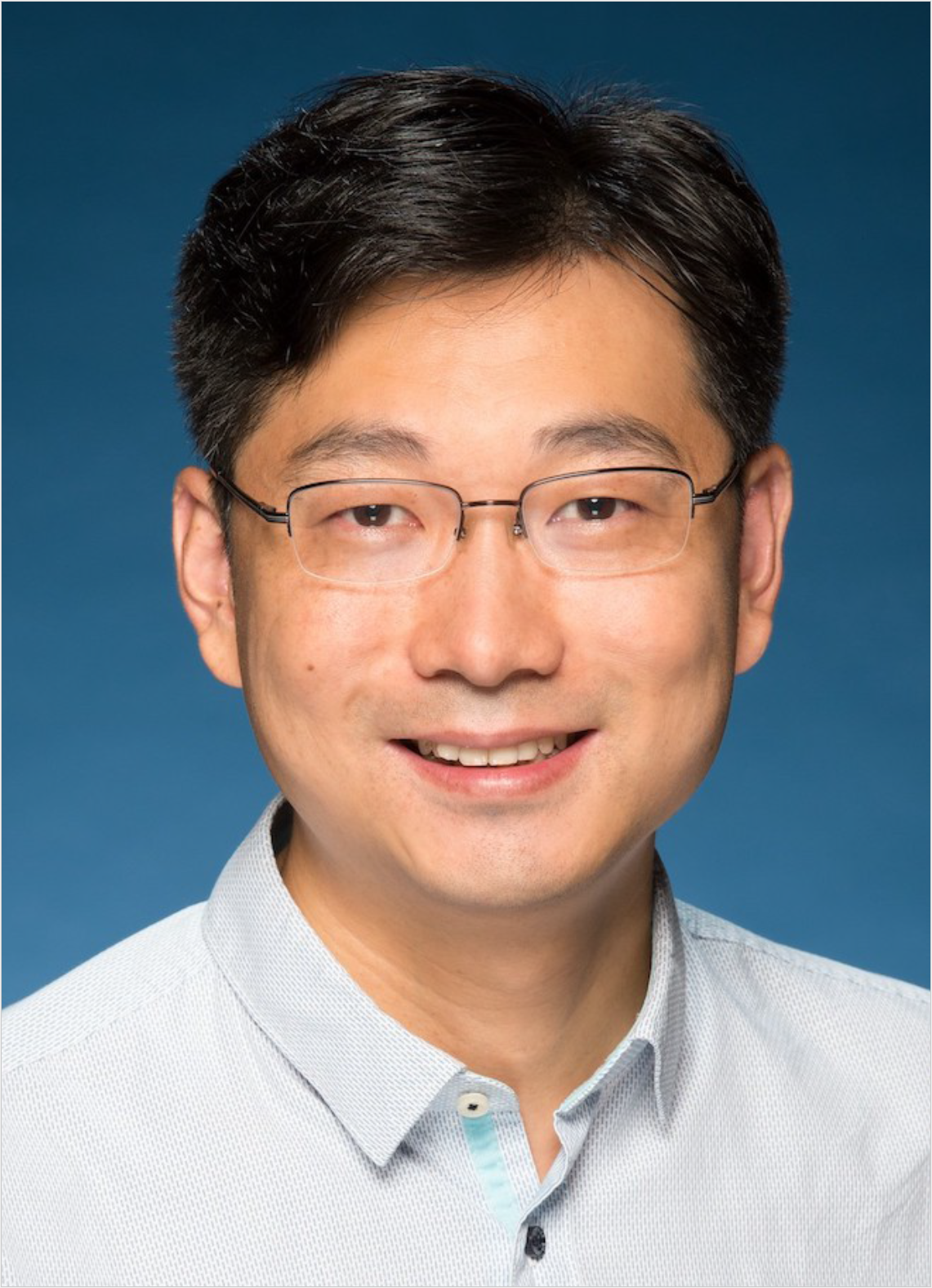}}]
{Xiaowen Chu}
(Fellow, IEEE) received the B.Eng. degree in computer science from Tsinghua University, Beijing, China, in 1999, and the Ph.D. degree in computer science from The Hong Kong University of Science and Technology, Hong Kong, in 2003. He is currently a Full Professor and Head of Data Science and Analytics Thrust at The Hong Kong University of Science and Technology (Guangzhou). His research interests include GPU Computing, Distributed Machine Learning, Cloud Computing, and Wireless Networks. He has won six Best Paper Awards at different international conferences, including IEEE INFOCOM 2021. He has published over 240 research articles at international journals and conference proceedings. He has served as an associate editor or guest editor of IEEE Transactions on Cloud Computing, IEEE Transactions on Network Science and Engineering, IEEE Transactions on Big Data, IEEE IoT Journal, IEEE Network, IEEE Transactions on Industrial Informatics, etc. 
\end{IEEEbiography}

\begin{IEEEbiography}
[{\includegraphics[width=1in,height=1.25in,clip,keepaspectratio]{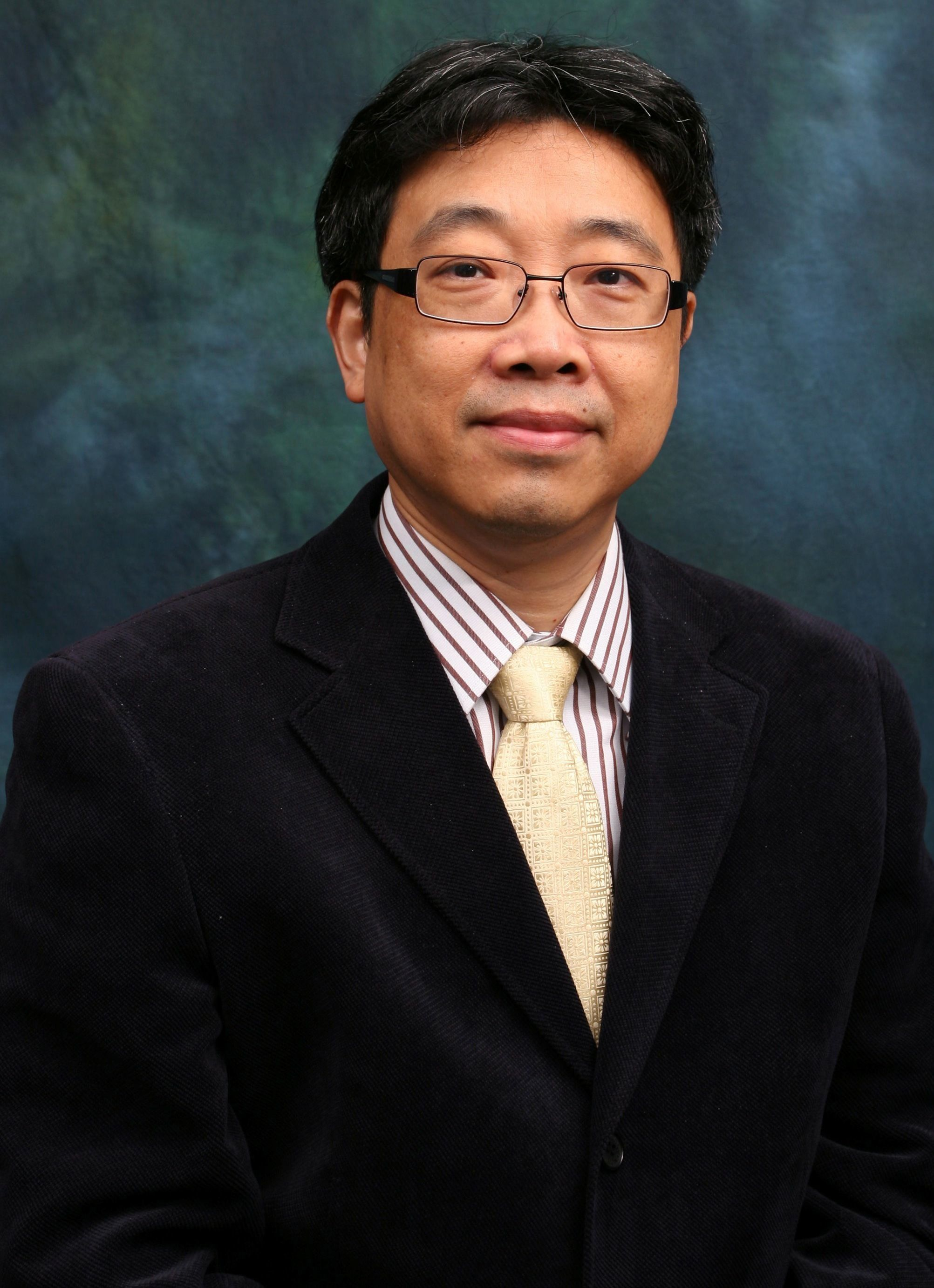}}]
{Jiannong Cao}
(Fellow, IEEE) received the B.Sc. degree in computer science from Nanjing University, China, in 1982, and the M.Sc. and Ph.D. degrees in computer science from Washington State University, USA, in 1986 and 1990, respectively. He is currently the Otto Poon Charitable Foundation Professor in data science and the Chair Professor of distributed and mobile computing with the Department of Computing, Hong Kong Polytechnic
University, Hong Kong. He is also the Director of the Internet and Mobile Computing Lab at the Department and the Associate Director of the University Research Facility in big data analytics. His research interests include parallel and distributed computing, wireless networks and mobile computing, big data and cloud computing, pervasive computing, and fault-tolerant computing. He has coauthored five books in Mobile Computing and Wireless Sensor Networks, co-edited nine books, and published over 600 papers in major international journals and conference proceedings. He is a Distinguished Member of ACM and a Senior Member of the China Computer Federation (CCF).
\end{IEEEbiography}

% \newpage
\appendices
\section{Tables}
The Tables \ref{table1} and \ref{table2} of this paper are listed in this appendix.
\begin{table*}[!t]
    \caption{Comparison of ours and current representative methods}
    \centering
    \resizebox{\textwidth}{!}{
    \begin{tabular}{m{2.0cm}<{} m{2.5cm}<{\centering}m{2.5cm}<{\centering}m{2.0cm}<{\centering}m{2.5cm}<{\centering}m{2.5cm}<{\centering}}
        \toprule
        Method & Cloud-enabled AIGC Service & Edge-enabled AIGC Service & Single-ES Offloading & Fixed Multi-ES Offloading & Adaptive Multi-ES Offloading\\
        \midrule
            Van \textit{et al.} \cite{van2023chatgpt} &\checkmark &\ding{53} &\ding{53} &\ding{53} &\ding{53} \\
            Du \textit{et al.} \cite{du2024exploring} &\ding{53} &\checkmark &\checkmark &\ding{53} &\ding{53} \\
            Xu \textit{et al.} \cite{xu2024dynamic} &\ding{53} &\ding{53} &\ding{53} &\checkmark &\ding{53} \\
            Ours &\ding{53} &\checkmark &\ding{53} &\ding{53} &\checkmark \\
        \bottomrule
    \end{tabular}}
    \label{table1}
\end{table*}

\begin{table}[!h]
    \caption{Commonly used notations in this paper}
    \centering
    \resizebox{\linewidth}{!}{
	\begin{tabular}{p{30pt}<{\centering}p{190pt}}
	  \toprule
	Symbol & Description\\
	\midrule
  	$\mathcal{T}$ & The set of time slot\\
	$\mathcal{B}$ & The set of ESs or BSs\\
	$\mathcal{N}_{b,t}$ & The task set that arrives to BS $b\in\mathcal{B}$ at time slot $t\in \mathcal{T}$\\
	$\Delta$ & The length of each time slot $t \in \mathcal{T}$\\
        $\tau_{n}$ & The deadline of processing task $n \in \mathcal{N}_{b,t}$\\
	$d_{n}$ & The size of task $n \in \mathcal{N}_{b,t}$ \\
        $\rho_{n}$ & The computation density of task $n \in \mathcal{N}_{b,t}$\\ 
        $\boldsymbol{x}$ & The decision variable of task workload allocation\\
	$\boldsymbol{v}$ & The transmission rate between BSs\\
        $f_{b'}$ & The computing capacity of ES $b' \in \mathcal{B}$ \\
        $T_{b,n,t,b'}^{\text{proc}}$ & The processing delay of task $n\in \mathcal{N}_{b,t}$ that is allocated from BS $b\in\mathcal{B}$ to ES $b'\in\mathcal{B}$ at time slot $t\in \mathcal{T}$\\
        $T_{b,n,t,b'}^{\text{wait}}$ & The waiting time of task $n$ in the transmission queue at local BS $b$ and the processing queue at ES $b'$\\
        $T_{b,n,t}^{\text{make}}$ & The offloading make-span of task $n\in \mathcal{N}_{b,t}$ that arrives to a BS $b\in\mathcal{B}$ at time slot $t\in \mathcal{T}$\\
        $q_{t,b'}$ & The processing queue workload length of ES $b'$ at the end of time slot $t-1$\\
        $q^{\text{tran,bef}}_{b,n,t}$ & The data size of the transmission queue at ES $b$ before sending out task $n$ at time slot $t$\\
        $q^{\text{proc,bef}}_{n,t,b'}$ & The workload length of the processing queue at ES $b'$ before receiving task $n$ at time slot $t$.\\
	\bottomrule
	\end{tabular}}
	\label{table2}
\end{table}

\section{Discussion}
Compared to our previous conference version \cite{xu2024enhancing}, we further discuss the rationale, scalability, and limitations behind our AMCoEdge method.

\subsubsection{Rationale} This paper presents a novel AMCoEdge method. The AMCoEdge performs online adaptive multi-ES collaboration computing to optimize the task offloading delay for AIGC services. The experimental results presented in Sections \ref{evaluation} and \ref{implementation} demonstrate that the AMCoEdge significantly outperforms the state-of-the-art methods in terms of the average offloading make-span and failure rate. The main reason for this improvement is the AMCoEdge, which adaptively leverages the computational capabilities of multiple ESs to accelerate task processing. Additionally, the AMCoEdge allocates task workloads to match all the ESs with idle computing resources. In contrast, existing collaborative MEC methods (e.g., \cite{fan2024collaborative, zhou2022two, xu2024dynamic,huang2025drmq}) only select fixed ESs for task acceleration. These methods often prolong the overall make-span of task offloading due to the under-utilization of ES resources.

On the other hand, compared to our previous conference version \cite{xu2024enhancing}, we propose the AMCoEdge-H method by designing a more efficient HECWA algorithm for the workload allocation stage. Although the proposed HECWA algorithm has little effect on the task offloading make-span and failure rate by neglecting the waiting time of task processing during workload allocation, it significantly improves the time complexity of the workload allocation algorithm and the running time of the overall AMCoEdge algorithm, as presented in Subsection \ref{cost-effectiveness subsection}, especially for large ES computation capacity.

\subsubsection{Scalability}
Our AMCoEdge method exhibits scalability. In particular, by retraining our model to suit new scenarios, the AMCoEdge can be effectively extended for a wide range of modern IoT applications. For instance, our AMCoEdge method can be applied to many real-world AIGC applications, such as text-to-image, image-to-image (as depicted in Figs. \ref{Fig1} and \ref{Fig11}), and text-to-video tasks.

\subsubsection{Limitation}
Our AMCoEdge method may have limitations that can be improved. For example, the pricing and energy consumption models are also important issues in modern IoT applications \cite{chen2024task}. It should be noted that AMCoEdge does not account for the different pricing and energy-consumption models among AIGC service providers.

\section{Proof of the Theorem \ref{NP-hard-theorem}}\label{NP-hard-theorem-proof}
\begin{proof}
The offline counterpart of the AM-CMEC problem is proven to be NP-hard through reducing the multi-knapsack problem \cite{gu2021layer} to it. Firstly, in the AM-CMEC  problem, all the task workloads $\{d_{n} \cdot \rho_{n} \cdot x_{b,n,t,b'} \}$ are viewed as candidate objects by relaxing the $\boldsymbol{x}$ as integers. The set $\mathcal{B}$ of ESs is viewed as a multi-knapsack. Hence, the task workloads assigned to the ES set in the AM-CMEC problem are viewed as candidate objects distributed across multiple knapsacks. The weight capacity of knapsack $b' \in \mathcal{B}$ is $f_{b'}$. More task workloads with lower make-spans are placed into the ES set $\mathcal{B}$ under the constraint of weight capacity $\boldsymbol{f}$. Thus, the AM-CMEC  problem is a multi-knapsack problem. Secondly, as a general case of the knapsack problem that is well-known NP-hard, the multi-knapsack problem is also an NP-hard problem. Therefore, the offline counterpart of the AM-CMEC problem is NP-hard.
\end{proof}

\section{Proof of the Theorem \ref{CWA-complexity-Theorem}}\label{CWA-complexity-Theorem-proof}
\begin{proof}
As shown in Algorithm \ref{CWA-procedure}, there is one ``for'' loop. In this loop, the time complexity of execution is $O(|\mathcal{L}^{b,n}_t|) \leq O(B) $. Then, if $k = 1$, the time complexity of execution is $O(1)$. Otherwise, the time complexity of execution is $O(|\mathcal{L}^{b,n}_t|)$ + $O(k^{3})$ which is the time complexity of solving the equations \cite{huang2022throughput} (Line 12). Furthermore, when $k$ > 1, the total time complexity of Algorithm \ref{CWA-procedure} is $O(k^3)$ since $|\mathcal{L}^{b,n}_t| \ll O(k^{3})$.
\end{proof}

\section{Proof of the Theorem \ref{effectiveness-Theorem}}\label{effectiveness-Theorem-proof}
\begin{proof}
This theorem can be proved by mathematical induction as follows. Firstly, when $k$ = 1, one ES is only selected to process the entire task $n$. We assume that the index of selected ES is $b'$, i.e., $\mathcal{L}^{b,n}_{t} = \{b'\}$. We have $x_{b,n,t,b'} = 1$, i.e., when $k$ = 1, the Eqn. (\ref{closed-form-formula}) is true. Secondly, when $k$ = 2, two ESs are only selected to process the entire task $n$ collaboratively. We assume that the indexes of selected ES are $b'$ and $i$, i.e., $\mathcal{L}^{b,n}_{t} = \{b', i\}$. Then, the $x_{b,n,t,b'}$ and $x_{b,n,t,i}$ are optimally achieved by the Eqn. (\ref{Eq-x-b-solution}) and (\ref{Eq-x-i-solution}). Obviously, when the $T^{\text{wait}}_{b,n,t,b'}$ and $T^{\text{wait}}_{b,n,t,i}$ equal to 0, the Eqn. (\ref{Eq-x-b-solution}) and (\ref{Eq-x-i-solution}) satisfy the Eqn. (\ref{closed-form-formula}). Therefore, when $k$ = 2, the Eqn. (\ref{closed-form-formula}) is also true. Thirdly, we prove that assuming $k$ = $m$ is true for the Eqn. (\ref{closed-form-formula}) and then $ k = m+1$ is also true for the Eqn. (\ref{closed-form-formula}) as follows. In particular, when $k$ = $m$ ($3 \leq m \leq B$), the $m$ ESs are selected to collaboratively process the entire task $n$. Here, we assume that the indexes set of the selected ES is $\mathcal{M}$, i.e., $\mathcal{L}^{b,n}_{t} = \mathcal{M}$. We have
\begin{equation*}
    x_{b,n,t,b'} = \frac{\prod_{u\in\mathcal{M}\setminus b'} T_{b,n,t,u}}{\sum_{v\in\mathcal{M}}\prod_{u\in\mathcal{M}\setminus v} T_{b,n,t,u}}.
\end{equation*}
Furthermore, when $ k = m + 1$, $m + 1$ ESs are simultaneously selected to process the entire task $n$ collaboratively. Assuming the index set $\mathcal{L}^{b,n}_{t}$ of selected ES includes $\mathcal{M}$ and $j$, i.e., $\mathcal{L}^{b,n}_{t} = \mathcal{M} \cup \{j\}$. We have 
\begin{align}
    x_{b,n,t,b'} & = \frac{\prod_{u\in\mathcal{M}\cup \{j\}\setminus b'} T_{b,n,t,u}}{\sum_{v\in\mathcal{M}\cup \{j\}}\prod_{u\in\mathcal{M}\cup \{j\}\setminus v} T_{b,n,t,u}} \nonumber\\
    & = \frac{\prod_{u\in\mathcal{L}^{b,n}_{t}\setminus b'} T_{b,n,t,u}}{\sum_{v\in\mathcal{L}^{b,n}_{t}}\prod_{u\in\mathcal{L}^{b,n}_{t}\setminus v} T_{b,n,t,u}},\nonumber
\end{align}
which shows that when $k$ = $m+1$, the Eqn. (\ref{closed-form-formula}) is also true. Finally, by the above derivation, the Eqn. (\ref{closed-form-formula}) is proved.
\end{proof}

\section{Proof of the Theorem \ref{HECWA-complexity-Theorem}}\label{HECWA-complexity-Theorem-proof}
\begin{proof}
As shown in Algorithm \ref{HECWA-algorithm}, for each ``for'' loop, the time complexity of execution is $O(k)$. Therefore, the time complexity of Algorithm \ref{HECWA-algorithm} is $O(k)$ + $O(k)$ = $O(k)$.
\end{proof}

\section{Proof of the Theorem \ref{AMCoEdge-complexity-Theorem}}\label{AMCoEdge-complexity-Theorem-proof}
\begin{proof}
As shown in Algorithm \ref{AMCoEdge-algorithm}, for each time slot $t \in \mathcal{T}$, there are two ``for'' loops ordered. In the first loop, since all BSs are executed in parallel, the time complexity is determined by the number $N_{b,t}$ of task and the solving time of ES selection and workload allocation. Furthermore, the time required to solve the ES selection problem is determined by the AEQN inference time, which is very quick since the AEQN model is pre-trained. Then, the time complexity of ES selection is approximately $O(1)$. The solving time of workload allocation is $O(k^3)$ or $O(k)$ as proved in \textbf{Theorem} \ref{CWA-complexity-Theorem} or \ref{HECWA-complexity-Theorem}. Thus, the maximal time complexity of the first ``for'' loop is $O(N_{b,t} \cdot (1+k^{3}))$ = $O(N_{b,t} \cdot k^{3})$. In the second loop, the time complexity is $O(1)$ because it is executed in parallel. Thus, the total time complexity of Algorithm \ref{AMCoEdge-algorithm} is also $O(N_{b,t} \cdot k^{3})$ for each time slot $t$. However, since $1 \leq k \leq B$ and $k$ is usually a small number, the time complexity of Algorithm \ref{AMCoEdge-algorithm} is approximately equal to $O(N_{b,t})$.
\end{proof}

\end{document}